\documentclass[10pt, letterpaper]{IEEEtran}
\usepackage[colorlinks,urlcolor=blue,linkcolor=blue,citecolor=blue]{hyperref}
\usepackage{array}
\usepackage[hang]{footmisc}

\usepackage[subrefformat=parens,labelformat=parens,caption=false,font=footnotesize]{subfig}
\usepackage{amsmath}
\usepackage[nolist,nohyperlinks]{acronym}
\usepackage{amsmath}

\begin{acronym}
    \acro{4G}{fourth generation}
    \acro{5G}{fifth generation}
    \acro{BPP}{binomial point process}
\acro{CDF}{cumulative density function}
    \acro{CF}{closed-form}
    \acro{ECDF}{empirical cumulative distribution function}
    \acro{GPS}{global positioning system}
    \acro{GE}{group exploration}
    \acro{GNSS}{global navigation satellite system}
    \acro{LCB}{lower confidence bound}
    \acro{LLR}{log-likelihood ratio}
    \acro{MAB}{multi-armed bandit}
    \acro{MBS}{macro base station}
    \acro{MEC}{mobile-edge computing}
    \acro{mIoT}{massive internet of things}
    \acro{MIMO}{multiple input multiple output}
    \acro{mm-wave}{millimeter wave}
    \acro{mMTC}{massive machine-type communications}
    \acro{MS}{mobile station}
    \acro{MVUE}{minimum-variance unbiased estimator}
    \acro{NLOS}{non line-of-sight}
    \acro{OFDM}{orthogonal frequency division multiplexing}
    \acro{PAC}{probably approximately correct}
    \acro{PDF}{probability density function}
    \acro{PGF}{probability generating functional}
    \acro{PLCP}{Poisson line Cox process}
    \acro{PLT}{Poisson line tessellation}
    \acro{PLP}{Poisson line process}
    \acro{PPP}{Poisson point process}
    \acro{PV}{Poisson-Voronoi}
    \acro{QoS}{quality of service}
    \acro{RAT}{radio access technique}
    \acro{RL}{reinforcement-learning}
    \acro{RSSI}{received signal-strength indicator}
    \acro{BS}{base stations}
    \acro{IoT}{internet of things}
    \acro{IIoT}{industrial internet of things}
    \acro{RF}{radio frequency}
    \acro{WSN}{wireless sensor node}
    \acro{SWIPT}{simultaneous wireless information and power transfer}
    \acro{UE}{user equipment}  
    \acro{wpt}{wireless power transfer}
    \acro{SINR}{signal to interference plus noise ratio}
    \acro{SNR}{signal to noise ratio}
    \acro{SIR}{signal to interference ratio}
    \acro{JSP}{joint success probability}
    \acro{PAoI}{peak age of information}
    \acro{AoI}{age of information}
    \acro{EH}{energy harvesting}
    \acro{UB}{upper bound}
    \acro{LB}{lower bound}
    \acro{AC}{actual case}
    \acro{L}{linear}
    \acro{NL}{non linear}
    \acro{P}{preemptive}
    \acro{NP}{non-preemptive}
    \acro{L EH}{linear energy harvesting}
    \acro{NL EH}{non linear energy harvesting}
\end{acronym}

\usepackage{graphicx,wrapfig,lipsum,subcaption}
\usepackage{rotating}
\usepackage{amsmath, amssymb, amsthm}
\usepackage{stmaryrd}
\usepackage{tikz}
\usepackage{verbatim}
\usepackage{url}
\usepackage[utf8]{inputenc}
\usepackage[english]{babel}
\newtheorem{theorem}{Theorem}[]

\newtheorem{lemma}[]{Lemma}

\usepackage{multicol}

\usepackage{algorithm}
\usepackage{algpseudocode}
\usepackage{scalerel}
\usepackage{multicol}
\usepackage{setspace}
\newtheorem{definition}{Definition}
\usepackage[noadjust]{cite}
\usepackage{booktabs}

\usepackage{graphicx}
\usepackage{subcaption}
\usepackage{bbm}
\usepackage[normalem]{ulem}
\usepackage{color}
\usepackage{dsfont}
\usepackage{bm}
\usepackage{setspace}
\usetikzlibrary{automata}
\usetikzlibrary{arrows}
\usetikzlibrary{shapes.geometric, arrows}
\usepackage[]{footmisc}
\usetikzlibrary{positioning}
\usetikzlibrary{calc}
\usetikzlibrary{arrows}
\hypersetup{nolinks=true}
\allowdisplaybreaks

\usepackage{caption}
\usepackage{mathtools}
\begin{document}
\title{Analysis of Age-Energy Trade-off in IoT Networks Using Stochastic Geometry}
\author{\IEEEauthorblockN{Songita Das, \emph{Graduate Student Member}, \emph{IEEE} and Gourab Ghatak, \emph{Member}, \emph{IEEE}}\\
\thanks{The preliminary version of this paper was presented at the 22nd International Symposium on Modeling and Optimization in Mobile, Ad hoc, and Wireless Networks. WiOpt, RAWNET workshop, 2024 ~\cite{10778324}. Songita Das is with the Bharti School of Telecommunication Technology and Management, Indian Institute of Technology Delhi, New Delhi 110016, India (email: bsz228102@iitd.ac.in). Gourab Ghatak is with the Department of Electrical Engineering, Indian Institute of Technology Delhi, New Delhi 110016, India (email: gghatak@ee.iitd.ac.in).}
}
\markboth{}
{Shell \MakeLowercase{\textit{et al.}}: A Sample Article Using IEEEtran.cls for IEEE Journals}
\maketitle
\begin{abstract}
We study an internet of things (IoT) network where devices harvest energy from transmitter power. IoT devices use this harvested energy to operate and decode data packets. We propose a slot division scheme based on a parameter $\xi$, where the first phase is for energy harvesting (EH) and the second phase is for data transmission. We define the joint success probability (JSP) metric as the probability of the event that both the harvested energy and the received signal-to-interference ratio (SIR) exceed their respective thresholds. We provide lower and upper bounds of (JSP), as obtaining an exact JSP expression is challenging. Then, the peak age-of-information (PAoI) of data packets is determined using this framework. Higher slot intervals for EH reduce data transmission time, requiring higher link rates. In contrast, a lower EH slot interval will leave IoT devices without enough energy to decode the packets. We demonstrate that both non-preemptive and preemptive queuing disciplines may have the same optimal slot partitioning factor for maximizing the JSP and minimizing the PAoI. For different transmit powers and deployment areas, we recommend the optimal slot partitioning factor for the above two metrics under both queuing disciplines.
\end{abstract}


\section{Introduction}
\subsection{Context and Motivation}
\IEEEPARstart{T}{he} \ac{AoI} is a metric used in communication networks to assess how quickly information reaches its destination in relation to when it was created at the source~\cite{Yates}. Contrary to the classical wireless performance metrics such as throughput and delay that characterize data transfer efficiency, the \ac{AoI} adequately takes into account the freshness of data. This makes it attractive for system design in real-time \ac{IoT} applications~\cite{Kaul}. Traditional performance metrics include latency and reliability, where latency refers to the time taken for data to travel from the source to the destination, while reliability measures the likelihood of delivering data without errors or loss. However, even in systems with low latency and high reliability, the \ac{AoI} can be high, indicating that the data, though transmitted quickly and accurately, may be outdated by the time it is received. For example, in autonomous vehicle communication \cite{8322644} and drone swarm coordination \cite{8647634}, vehicles and drones communicate data regarding speed and position to ensure safety. Nonetheless, if updates occur infrequently and conditions fluctuate rapidly, even low-latency and reliable data may become obsolete, leading to a high \ac{AoI}. For instance, if a drone alters its position or a car abruptly brakes after transmitting an update, subsequent drones or vehicles may receive obsolete information, jeopardizing safety and heightening the risk of collisions or mission failure, despite the system’s low latency and high reliability. It is to be noted that in systems featuring multiple devices transmitting simultaneously, interference deteriorates communication, resulting in missed or erroneous updates that may cause collisions or other significant problems. Conversely, noise generally diminishes signal quality without entirely obstructing communication. Commonly employed downlink techniques in drone swarm coordination are dedicated ground control systems, which act as the central hub for delivering commands and data to drones. In situations where drones communicate via cellular networks, a central control system or cloud server transmits real-time updates and commands via the cellular network. Likewise, in satellite communication, the central control system uses a satellite link to send commands or data to drones thereby facilitating long-range communication.

Additionally, the devices for such IoT applications are often battery-constrained, and regular battery replacements are infeasible. To alleviate this \ac{EH} is a key solution that converts ambient energy into usable electrical power~\cite{Lu}. \ac{EH} provides a supplemental or alternative power source, allowing devices to run independently and sustainably~\cite{Kumar}. Since the wireless signals are repurposed for \ac{EH}, often it has an adverse impact on the decoding of the received signals, and consequently, the \ac{AoI}. Hence, the joint characterization and optimization of \ac{EH} and \ac{AoI} in wireless networks is critical and is the focus of this paper. 
\vspace{-0.3cm}
\subsection{Literature Review} 
Abd-Elmagid {\it et al.}~\cite{Abd-Elmagid1} provide a comprehensive overview of \ac{AoI} and its variants and their utility in designing freshness-aware IoT networks. They contrast the key differences of \ac{AoI} with respect to classical wireless network performance metrics such as throughput and delay. Furthermore, they reveal achievable \ac{AoI} regions with optimal energy sampling and transmission policies. Yates and Kaul~\cite{Yates2} analyzed \ac{AoI} in queues for stochastic hybrid systems that utilize memory-less service. Indeed, \ac{AoI} and its optimization have been the subject of extensive research in a variety of network scenarios in order to improve the freshness of data. Chen {\it et al.}~\cite{Chen} studied multi-hop wireless sensor nodes (WSNs) that employ \ac{EH} while optimizing \ac{AoI}. In particular, they formulate a problem that is NP-hard in order to minimize the same. Consequently, they propose an algorithm \texttt{MAoIG} and characterize the bounds on \ac{PAoI} and average AoI achieved by their algorithm. On the contrary, Zhu {\it et al.}~\cite{Zhu} focused on the scheduling challenges associated with minimizing \ac{AoI} in battery-free one-hop WSNs by employing a variety of \ac{EH} strategies. In particular, the authors discuss the optimal offline policy and introduce competitive online policies for minimizing \ac{AoI}. Zhang {\it et al.}~\cite{Zhang} investigated the impact of deploying mobile edge computing (MEC) servers on the computational ability of the sensor nodes. They formulated a non-convex problem to optimize the sampling rate, computation rate, and the \ac{AoI}, and proposed a successive convex approximation technique to solve the same. Kim~{\it et al.}~\cite{Kim} investigated optimal duty cycles for self-sustaining operations. For further works on joint \ac{AoI} and \ac{EH} optimization, we refer the reader to the references~\cite{Hentati, Arafa, Yao, Costa, wang, Sun, Sinha, Hirosawa}.

Rigorous optimization procedures, although useful for specific network instances, do not provide insights into the variance of performance measures across different network realizations. In this regard, stochastic geometry provides useful tools for the statistical characterization of such networks, thereby enabling the derivation of key system-design insights and dimensioning rules. Leveraging stochastic geometry, Mankar {\it et al.}~\cite{Mankar} study both preemptive and non-preemptive queuing and their impact on the \ac{AoI} in large-scale wireless networks constituting source-destination pairs. In particular, they derive tight bounds on the moments and spatial distribution of \ac{PAoI} using a bipolar Poisson point process model, which are validated through numerical results. Abd-Elmagid {\it et al.}~\cite{Abd-Elmagid2} investigate the joint transmission coverage and energy harvesting performance in a network with node locations modeled as a Poisson cluster process (PCP). They state that computing the distribution of shot noise processes associated with the PCP is challenging. Consequently, they prescribe approximations for the characterization of joint signal and energy coverage. Yang {\it et al.}~\cite{Yang22} have analyzed the spatio-temporal performance of \ac{AoI} by using queuing theory in conjunction with stochastic geometry. They showed that a decentralized scheduling policy reduces \ac{AoI} through local observations. The suggested method adjusts radio access probabilities to accommodate traffic changes to reduce \ac{PAoI}, and accommodate network expansion. Furthermore, they extended their investigation in \cite{Yang} to show that the last-come, first-served (LCFS) order has a variable impact on \ac{AoI} in relation to the deployment density, while the slotted ALOHA protocol does not reduce \ac{AoI} at low packet arrival rates.

Very few studies \cite{Sleem, Abd-Elmagid3}  have investigated \ac{EH} and \ac{AoI} jointly in uplink stochastic geometry frameworks, where source nodes perform \ac{EH}. In \cite{Sleem}, the authors have characterized the \ac{AoI} performance of the users by considering their \ac{EH} capabilities. Based on their formulation, they have optimized the channel access policy that minimizes the \ac{AoI}. On the contrary, \cite{Abd-Elmagid3} has studied the distributional properties of \ac{AoI} by assuming that the status packets and the energy packets arrive at the devices following a Poisson arrival process. This enhances our understanding of the statistics of the \ac{AoI} along with \ac{EH} capabilities. However, none of the existing works derive the efficacy of the \ac{EH} process jointly with the \ac{AoI}. In particular, when fewer transmit nodes are located in the network, this not only enhances the signal coverage due to limited interference but also reduces the ability of the users to effectively harvest energy. In order to jointly study this, we introduce the term called \ac{JSP} and investigate how the temporal resources can be partitioned in order to maximize either the \ac{JSP} or minimize the \ac{PAoI}.

\subsection{Contributions}
We consider an energy-harvesting \ac{IoT} network where the \ac{IoT} devices operate with harvested energy from the transmitters of the network. Given that the devices have harvested sufficient energy, they attempt to decode the downlink transmissions from the \ac{BS}. Based on this system, we make the following contributions:
\begin{enumerate}
\item  First, we study \ac{NL EH} based on which we derive the probability of the joint event that the harvested energy at the device exceeds an energy threshold and that the downlink \ac{SIR} exceeds an \ac{SIR} threshold. We term this as the \ac{JSP} and it determines the ability of any device in the network to simultaneously harvest the necessary energy to operate and the ability to decode the downlink data. To the best of our knowledge, this joint characterization has not yet been studied in stochastic geometry frameworks.
\item In order to allot resources between the \ac{EH} and the data transmission functionalities, we consider a slotted frame design. Next, based on the derived \ac{JSP} and by considering the slotted frame, we characterize the \ac{PAoI} at the IoT devices with respect to the frame partitioning factor $\xi$.
\item Interestingly, we show that for every given value of transmit power, there exists a unique optimal slot partitioning factor $\xi^{*}$ that, in both the preemptive and non-preemptive queuing disciplines, maximizes JSP while minimizing PAoI for a disc of a specific radius. Following this, we derive several system design insights. Based on the application of interest, the value of $\xi$ must be tuned appropriately. 
\end{enumerate}

\section{System Model and Problem Formulation} 

\subsection{Network Geometry}
We consider an \ac{IoT} network wherein the \ac{IoT} devices are served by their nearest transmitter. Let the locations of the \ac{IoT} devices be denoted by $s_{n}$, $n \in \{1, \ldots, N\}$. The transmitter locations are denoted by $r_{k}$. The typical \ac{IoT} device is denoted by $s_{1}$ while the location of the transmitter closest to the typical \ac{IoT} device is denoted by $r_{1}$. The transmitter locations are modeled as a \ac{PPP}, $\Phi$ with intensity $\lambda$ in a disk $\mathcal{B}(s_1, R)$ of radius $R$ centered at $s_1$. The distances between the typical \ac{IoT} device and the typical transmitter are represented by $d_{1}=\left\|s_{1}-r_{1}\right\|$ and the distances from the typical IoT device to any other transmitter are represented by $d_{k}=\left\|s_{1}-r_{k}\right\|$.

\subsection{Propagation Model}
Let the transmit power multiplied by the path-loss constant for all the transmitters be $P_{t}$. Accordingly, the power received by the typical IoT device from the associated transmitter in the $j-$th slot is given by $P_{R_{1}}=P_{t}g_{1}^{j}d_{1}^{-\alpha }$ while the interference power is given by $P_{R_{k}}= P_{t}\sum_{k=2}^{K}g_{k}^{j}d_{k}^{-\alpha}$. The small-scale fading gains are assumed to be exponentially distributed and denoted by $g_{1}^{j}$ for the typical transmitter and $g_{k}^{j}$ for the $k-$th interfering transmitter. Due to the \ac{PPP} assumption for the location of the interferers, $K$ is a Poisson distributed random variable with parameter $\lambda \pi R^2$. Assume that each time slot is of duration $\tau$, and is divided into two distinct phases- an \ac{EH} phase of duration $\xi \tau $ and a data transmission phase of duration $(1- \xi) \tau$. During the \ac{EH} phase, the \ac{IoT} device harvests RF energy from the downlink transmissions of all the transmitters. While, during the data transmission phase, the typical IoT device uses the harvested energy to decode the packet received from the typical transmitter. Let $E_{1}^{j}$ represent the energy harvested by the typical IoT device in the $j^{th}$ time slot. It has two components: the desired signal and the interfering signals. Mathematically,
\begin{align}
E_{1}^{j}&=\xi \tau P_{R_{1}}
=\eta\xi\tau P_{t}\left(g_{1}^{j}d_{1}^{-\alpha}+\sum_{k=2}^{K}g_{k}^{j}d_{k}^{-\alpha}\right).
\label{eq1}
\end{align}
where, $\eta$ represents energy efficiency. The typical IoT device must harvest more than $E_{th}$ amount of energy to be able to power its receiving circuitry and, hence, receive data successfully during the information reception phase.
In order to initiate the harvesting process, the input of any practical harvesting circuit must exceed the minimum power ($Pr_{min}$). Furthermore, the output of these harvesting circuits becomes saturated when the input power reaches the saturation threshold ($Pr_{max}$), resulting in a constant output despite fluctuations in input power. Therefore, the energy that is harvested in the non-linear scenario can be assessed as 
\begin{equation}
E_{1}^{j}\!\!=\!\!\begin{cases}
0 \!&,Pr<Pr_{min} \\
\eta \xi \tau P_{t}(g_{1}d_{1}^{-\alpha }\!\!+\!\!\sum_{k=2}^{K}g_{k}d_{k}^{-\alpha })\! &,Pr_{min}\!<\!Pr\!<\!Pr_{max} \\
\eta \xi \tau Pr_{max} \!&,Pr_{max}<Pr
\label{e2_lower}
\end{cases}
\end{equation}
where, $Pr_{min}=P_{t}(g_{1}d_{1}^{-\alpha }+d_{K}^{-\alpha }\sum_{k=2}^{K}g_{k})$, $Pr=P_{t}(g_{1}d_{1}^{-\alpha }+\sum_{k=2}^{K}g_{k}d_{k}^{-\alpha })$ and $Pr_{max}=P_{t}(g_{1}d_{1}^{-\alpha }+d_{1}^{-\alpha }\sum_{k=2}^{K}g_{k})$

Let $\Upsilon_{1}^{j}$ represent the \ac{SIR} measured at the typical IoT device and it is the ratio of the signal power to the interference power received at $s_1$ in the $j^{th}$ time slot. Mathematically, 
\begin{align}  
\Upsilon_{1}^{j}=\frac{P_{t}g_{1}^{j}d_{1}^{-\alpha}}{\sum_{k=2}^{K}P_{t}g_{k}^{j}d_{k}^{-\alpha}}.
\label{eq2}
\end{align}

In order to successfully decode the data packet, the \ac{SIR} at the typical \ac{IoT} device must be greater than an 
\ac{SIR} threshold $\beta$. Let the number of bits to be transmitted be given by $\sigma$ during the $(1-\xi)\tau$ faction of the slot dedicated for data transmission. Furthermore, let the required data rate, $r$ is given by $r =\frac{\sigma }{(1-\xi)\tau}$. From the Shannon-Hartley theorem, the maximum data rate is the channel capacity, i.e., $C = B \log_{2}(1+\Upsilon_{1}^{j})$, where $B$ is the bandwidth of the channel. Thus, for transmission success, we need $C > r$ or $\Upsilon_{1}^{j}>2^{\frac{r}{B}}-1$. Thus, we set our \ac{SIR} threshold as $\beta = 2^{\frac{r}{B}}-1$.

\subsection{Performance Metrics of Interest}
Our main metrics of interest are \ac{JSP} and \ac{PAoI} denoted by $\mu _{\phi }$ and $A_{i}$, respectively.
\begin{definition}
The \ac{JSP}, $\mu_{\phi}$ is defined as the probability of the joint event that the harvested energy exceeds the energy threshold $E_{th}$ and the downlink \ac{SIR} exceeds the \ac{SIR} threshold $\beta$. Mathematically,
\begin{equation}
\begin{aligned}
\mu _{\phi } = \mathbb{P}[E_{1}^{j}> E_{th},\Upsilon_{1}^{j}>\beta].
\label{alpha}
\end{aligned}
\end{equation}

\end{definition}
\begin{definition}
The \ac{AoI}, $\Delta(t)$ at time slot $t$ is defined as the time elapsed since the last received packet at the \ac{IoT} device was generated at the transmitter. Let  $G_{t_{i}}$ and $t_{i}$ denote the generation and reception time instants of the $i^{th}$ packet packet at the transmitter and the IoT device, respectively. Mathematically,
\begin{align}
    \Delta \left ( t+1 \right )=    
    \begin{cases}
    \Delta \left ( t \right )+1, & \text{ if transmission fails,} \\
t_{i}-G_{t_{i}}+1, & \text{ otherwise.}\\    
    \end{cases}
\end{align}
\end{definition}
Thus, $\Delta(t)$ increases in a staircase fashion with time slots and drops upon reception of a new packet at the destination, to the total number of slots experienced by this new packet in the system. The typical \ac{IoT} device experiences interference from other transmitters.
\begin{figure} 
\centering
\includegraphics[scale=0.35]{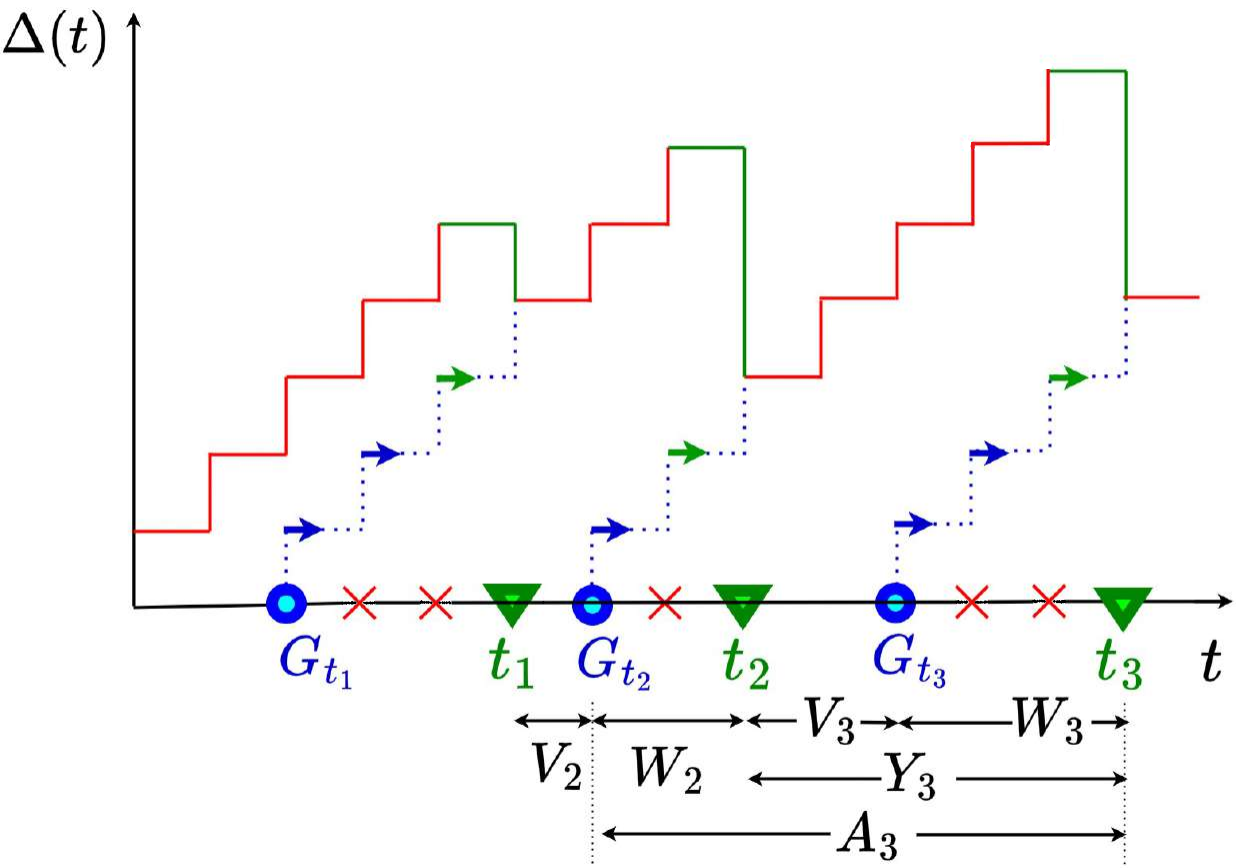}
\caption{Illustration of sample path of \ac{AoI} under non-preemptive queue discipline.}
\label{staircase_AoI}
\end{figure}
\subsection{AoI under Non-Preemptive Queue Discipline}
We assume a non-preemptive queuing model. Let the packets arrive at the transmitter following a Bernoulli process with probability $p_{a}$ at each time slot. On the arrival of a new packet, the transmitter attempts to send this packet to the \ac{IoT} device at each consecutive time slot. Until the transmission is successful, newly arrived packets are dropped. Let $W_i = t_{i}-G_{t_{i}}$ denote the time slots required for successful transmission of the $i-$th arrived packet. On successful delivery of the packet at the \ac{IoT} device, the transmitter admits new packets. Let $V_i$ denote the time slots between the successful delivery of the $(i-1)$-th packet and the arrival of the $i-$th packet. Thus, 
\begin{equation}
\begin{aligned}
Y_i = W_i + V_i,
\label{Yi}
\end{aligned}
\end{equation}
denotes the time elapsed between the successful delivery of the $(i-1)$-th and the $i$-th packets. In a given time slot, both the \ac{EH} and \ac{SIR} at the typical IoT must surpass the $E_{th}$ and ${\beta}$ thresholds simultaneously for packet transmission to be successful, the probability of which is determined by $\mu _{\phi }$. Thus, given the arrival of a packet for transmission, the success event at each time slot is independent and identically distributed. Accordingly, the transmissions on a typical link can be represented using a Geo/Geo/1 model.

Fig. \ref{staircase_AoI} depicts an \ac{AoI} process sample path. Crosses indicate instances of packets being dropped during server activity of $i^{th}$ packet trying to reach the destination. Then, the \ac{PAoI}, $A_{i}$, corresponding to the $i-$th packet measures the maximum time elapsed since the last received, i.e., the $(i-1)-${th} packet at the destination was generated. This is given by 
\begin{equation}
\begin{aligned}
A_{i} = W_{i-1}+ Y_{i}.
\label{e2}
\end{aligned}
\end{equation}
\begin{figure} 
\centering
\includegraphics[scale=0.20]{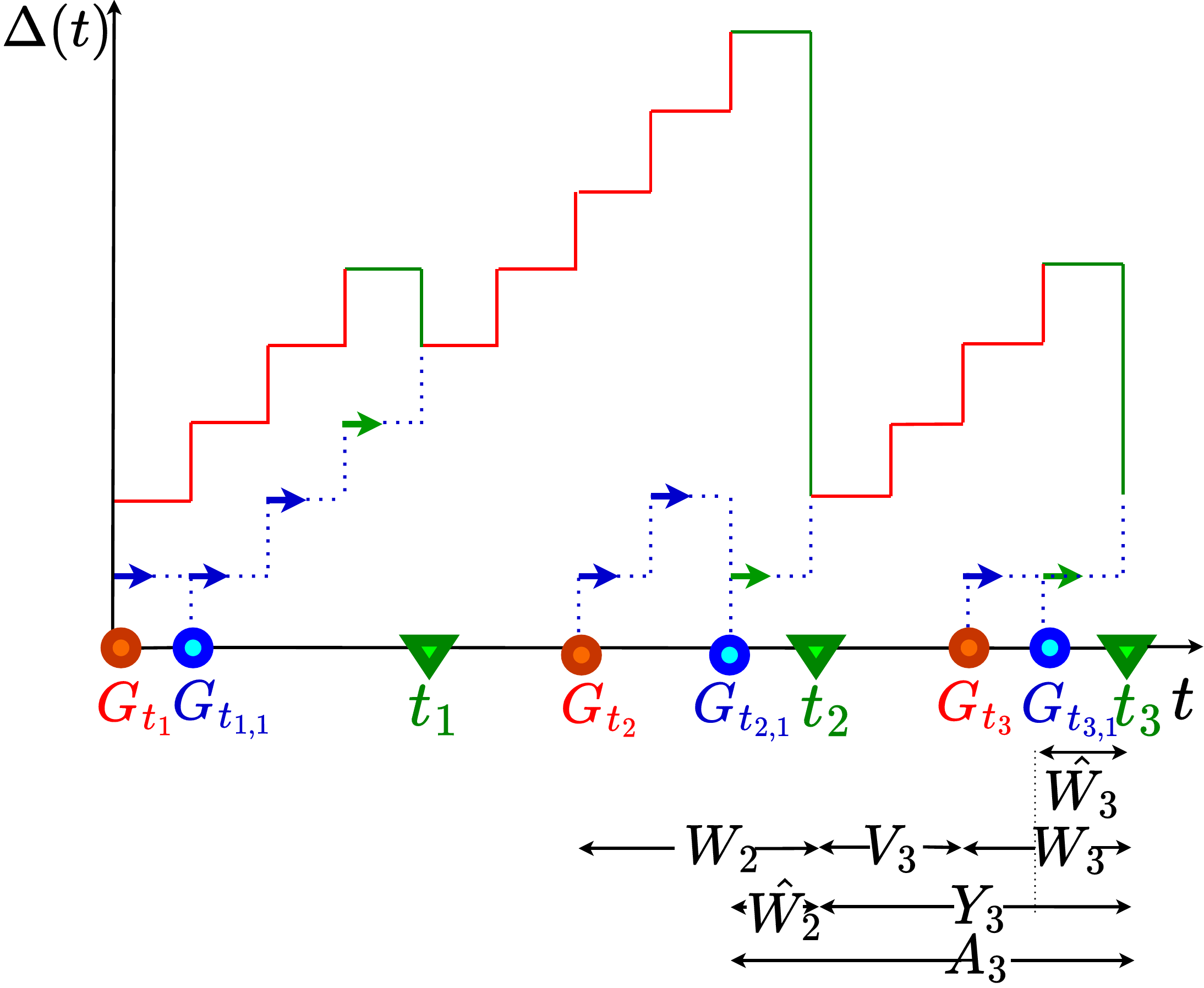}
\caption{Illustration of sample path of \ac{AoI} under preemptive queue discipline.}
\label{staircase_AoI_preemptive}
\end{figure}
\begin{figure} 
\centering
\includegraphics[scale=0.20]{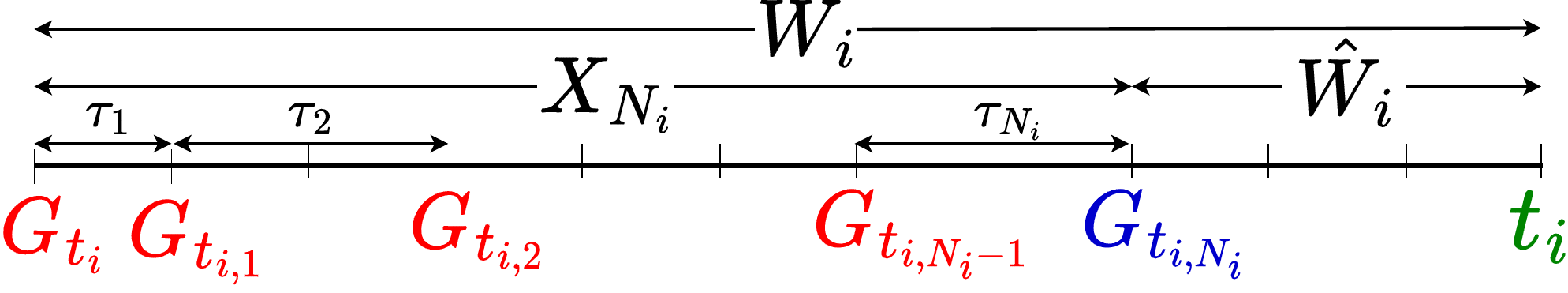}
\caption{Illustration of the process $\hat{W_{i}}$ under preemptive queue discipline.}
\label{staircase_AoI_preemptiveW}
\end{figure}
\vspace{-0.9cm}
\subsection{AoI under Preemptive Queue Discipline}
In preemptive discipline, each source transmits its most recent update in a given slot, reducing the \ac{AoI}, especially when updates arrive frequently but delivery is slow. This approach is optimal for minimizing \ac{AoI}, as the latest update is always transmitted. Fig. \ref{staircase_AoI_preemptive} illustrates the AoI process, where $G_{t_{i}}$ show update arrivals, $t_{i}$ show deliveries, and $G_{t_{i,n}}$ indicate when older updates are replaced by newer ones. This term marks the $n$-th replacement of the $i$-th update.

Let the packets arrive at the transmitter following a Bernoulli process with probability $p_{a}$ at each time slot. On the arrival of a new packet, the transmitter attempts to send this packet to the \ac{IoT} device at each consecutive time slot. If a new packet arrives at the transmitter, then the packet in service is dropped and the newly arrived packet starts getting serviced at each consecutive time slot. Let $W_{i} = t_{i}-G_{t_{i}}$ denote the time slots required for the successful reception of the $i$-th packet at the destination. $ \hat{W_{i}}$ denotes the number of time slots required to deliver the latest n-th replacement of the i-th packet at the destination and is given by $ \hat{W_{i}} = W_{i} - X_{N_{i}}$, where, $X_{N_{i}} = \sum_{n=1}^{N_{i}} \tau _{n}$ and $N_{i}$ is the number of new replacement packet arrivals occurring between the arrival and successful delivery of the $i$-th packet at the transmitter and destination \ac{IoT} device, respectively. Fig. \ref{staircase_AoI_preemptiveW} illustrates the process $\hat{W_{i}}$ under preemptive queue discipline. $\tau_{n}$ is the number of slots between the arrivals of two successive packet replacements that are i.i.d. and follow a geometric distribution with parameter $p_{a}$.

On successful delivery of the packet at the \ac{IoT} device, the transmitter admits new packets. Let $V_i$ denote the time slots between the successful delivery of the $(i-1)$-th packet and the arrival of the $i$-th packet. Thus, 
\begin{equation}
\begin{aligned}
Y_i = W_i + V_i,
\label{Yip}
\end{aligned}
\end{equation}
denotes the time elapsed between the successful delivery of the $(i-1)$-th and the $i$-th packets. In a given time slot, both the \ac{EH} and \ac{SIR} at the typical IoT must surpass the $E_{th}$ and ${\beta}$ thresholds simultaneously for packet transmission to be successful, the probability of which is determined by $\mu _{\phi }$. Thus, given the arrival of a packet for transmission, the success event at each time slot is independent and identically distributed. Accordingly, the transmissions on a typical link can be represented using a Geo/Geo/1 model.

Fig. \ref{staircase_AoI_preemptive} depicts an \ac{AoI} process sample path where we observe that the \ac{PAoI}, $A_{i}$, corresponding to the $i-$th packet measures the maximum time elapsed since the last received packet (i.e. the $(i-1)$-th packet received at $t_{i-1}$) at the destination, was generated (i.e. the n-th replacement of the $(i-1)$-th packet was generated at $G_{t_{i-1,n}}$) at the source. This is given by 
\begin{equation}
\begin{aligned}
A_{i} = \hat{W}_{i-1}+ Y_{i}.
\label{e2p}
\end{aligned}
\end{equation}
\vspace{-0.9cm}
\section{Joint Success Probability Analysis}
In this section, we analyze the performance of the proposed system model in terms of $\mu _{\phi }$. Since we focus on the impact of interference on the \ac{JSP}, in what follows, we consider that at least two transmitters are located within the region of interest. Then, the following result characterizes the distances to the nearest and farthest transmitter.
\begin{lemma}
Given that at least two points exist within the disc, the \ac{PDF} of the farthest point in the disc is
\begin{align}   \label{e9}
f_{d_{K}}(r)=\frac{2\lambda \pi r e^{-\lambda \pi \left ( R^{2}-r^{2} \right )}}{1-\left(1  + \lambda \pi R^2\right) e^{-\lambda \pi R^{2}}},  \quad r\geq 0, k\geq 2.
\end{align}
while, the PDF of the nearest point in the disc is as follows: 
\begin{align}
f_{d_{1}}(r)=\frac{2\lambda \pi r e^{-\lambda \pi r^{2}}}{1-\left(1  + \lambda \pi R^2\right) e^{-\lambda \pi R^{2}}} ,  r\geq 0.
\label{e10}
\end{align}
\end{lemma}
Let us recall that the \ac{JSP} is defined as the joint event that $E_{1}^{j}> E_{th}$ and $\Upsilon_{1}^{j}> \beta$. Substituting the value of the harvested energy and the \ac{SIR} in equation \eqref{alpha} we get,
\begin{equation}
\begin{aligned} 
\!\mu_{\phi}\!=\!\mathbb{P}
      \Bigg (\!\!g_{1}^{j}\!>\!d_{1}^{\alpha}\! \Bigg [\!\frac{E_{th}}{\eta\xi\tau P_{t}}\!-\!\sum_{k=2}^{K}\!g_{k}^{j}d_{k}^{-\alpha}\! \Bigg ],
g_{1}^{j}\!>\!\beta d_{1}^{\alpha}\!\sum_{k=2}^{K}\!g_{k}^{j}d_{k}^{-\alpha}\! \Bigg ).
\end{aligned}
\end{equation}
The characterization of the above is challenging since it involves the derivation of the distribution of the sum of weighted exponential terms of the form $\sum_{k=2}^{K}g_{k}^{j}d_{k}^{-\alpha}$. Thus, obtaining the closed-form expression of $\mu _{\phi }$ is infeasible. Consequently, we focus on deriving the upper bound $\mu _{\phi ,U}$ and the lower bound $\mu _{\phi ,L}$, of $\mu_{\phi }$. 
The idea behind the upper and the lower bounds of the \ac{JSP} is as follows. If all the interfering transmitters are placed at the farthest location $d_{k}$ in the disc, then energy harvested by the typical IoT is minimum. Similarly, if all the interfering transmitters are placed at the nearest location $d_{1}$ in the disc, then the energy harvested by the typical IoT is maximum. Let $E_{1}^{j, min}$ and $E_{1}^{j, max}$ denote the minimum and maximum harvested energy respectively. With the help of \eqref{eq1}, we can write them as follows: 
\begin{equation}
    \begin{aligned}  \label{e11_lower}
    E_{1}^{j, min}=\eta\xi\tau P_{t}\left(g_{1}^{j}d_{1}^{-\alpha}+d_{K}^{-\alpha}\sum_{k=2}^{K}g_{k}^{j}\right).
    \end{aligned}
\end{equation}
\begin{equation}
    \begin{aligned}   \label{e11_upper}
    E_{1}^{j, max}=\eta\xi\tau P_{t}\left(g_{1}^{j}d_{1}^{-\alpha}+d_{1}^{-\alpha}\sum_{k=2}^{K}g_{k}^{j}\right).
    \end{aligned}
\end{equation}
where $E_{1}^{j}$ is energy harvested by $s_{1}$ node in $\xi\tau$ time, which is then used to keep $s_{1}$ node active such that it is able to read the desired transmitted signal from $r_{1}$ to $s_{1}$ in the remaining $(1-\xi)\tau$ time. Similarly, if all the interfering transmitters are placed in the same circumference as the signal transmitter i.e., at the nearest distance $d_{1}$ to the typical IoT, then this causes maximum interference to the typical IoT, hence \ac{SIR} experienced by it is minimum and vice versa, as clear from the expression of \ac{SIR}, 
\begin{equation}
\begin{aligned}   \label{SII_min}
\Upsilon_{1}^{j, min}=\frac{P_{t}g_{1}^{j}d_{1}^{-\alpha}}{d_{1}^{-\alpha}{\sum_{k=2}^{K}P_{t}g_{k}^{j}}}.
\end{aligned}
\end{equation}
\begin{equation}
\begin{aligned}   \label{SII_max}
\Upsilon_{1}^{j, max}=\frac{P_{t}g_{1}^{j}d_{1}^{-\alpha}}{d_{K}^{-\alpha}{\sum_{k=2}^{K}P_{t}g_{k}^{j}}}.
\end{aligned}
\end{equation}
Case A: $Pr<Pr_{min}$.\\
Accordingly, we can write $\mu _{\phi ,L}$ and $\mu _{\phi ,U}$ as $\mu _{\phi ,L}=\mathbb{P}[(E_{1}^{j}> E_{th}) , (\Upsilon_{1}^{j,min}> \beta )] = 0$ and $\mu _{\phi ,U} = \mathbb{P}[(E_{1}^{j}> E_{th}),(\Upsilon_{1}^{j,max}> \beta )] = 0$.\\ \\
Case B: $Pr_{min}\!<\!Pr\!<\!Pr_{max} $.\\
Accordingly, we can write $\mu _{\phi ,L}$ and $\mu _{\phi ,U}$ as $\mu _{\phi ,L}=\mathbb{P}[(E_{1}^{j,min}> E_{th}) 
,    (\Upsilon_{1}^{j,min}> \beta )]$ and $\mu _{\phi ,U} = \mathbb{P}[(E_{1}^{j,max}> E_{th}),(\Upsilon_{1}^{j,max}> \beta )]$. The analytical expression of $\mu _{\phi ,L}$ and $\mu _{\phi ,U}$ are given in Theorem~\ref{theorem1} and Theorem~\ref{theorem2}, respectively.
\begin{theorem}   \label{theorem1}
The lower bound of $\mu _{\phi }$ is given by 
\begin{equation}
\begin{aligned}
\!\mu _{\phi ,L}\!\!=\!\! & \sum_{k = 2}^{\infty }\!\!\Big [\!\!\int_{d_{1}=0}^{R} \!\!\int_{d_{K}=d_{1}}^{R}\!\!\int_{z=0}^{\!\!\frac{E_{th}}{\eta \xi \tau P_{t}(\beta d_{1}^{-\alpha }\!+\!d_{K}^{-\alpha }\!)}}\!\!(\!\frac{z^{k-1}}{(k-1)!}\!)\!\exp (\frac{-E_{th}d_{1}^{ \alpha }}{\eta \xi \tau P_{t}}\\
& +d_{1}^{\alpha }d_{K}^{-\alpha }z-z 
 ) f_{d_{1}}(r)f_{d_{K}}(r)f_{K}(k)\Big ]d_{z} d_{d_{K}} d_{d_{1}}+ \\ 
 & \sum_{k = 2}^{\infty }\Big[\int_{d_{1}=0}^{R}\int_{d_{K}=d_{1}}^{R}
 \int_{z=\frac{E_{th}}{\eta \xi \tau P_{t}(\beta d_{1}^{-\alpha }+d_{K}^{-\alpha })}}^{\infty}(\frac{z^{k-1}}{(k-1)!})\\ 
 \nonumber
\end{aligned}
\end{equation}
\begin{equation}
\begin{aligned}
 & \exp(-\beta z-z)f_{d_{1}}(r)f_{d_{K}}(r)f_{K}(k)\Big ] d_{z} d_{d_{K}} d_{d_{1}}, 
\label{theorem1eq}
\end{aligned}
\end{equation}
where $f_{K}(k) = \frac{e^{-\lambda \pi R^{2}}(\lambda \pi R^{2})^{k}}{k!}$ and $f_{d_{1}}(r)$ and $f_{d_{K}}(r)$ are obtained from \eqref{e10} and \eqref{e9}.
\end{theorem}
\begin{proof}
A detailed proof is given in Appendix A.    
\end{proof}
\begin{theorem}   \label{theorem2}
The upper bound of $\mu _{\phi }$ is given by 
\begin{equation}
\begin{aligned}
\!\mu _{\phi ,U}\!\!=\!\! & \sum_{k = 2}^{\infty }\!\!\Big [\!\!\int_{d_{1}=0}^{R} \!\!\int_{d_{K}=d_{1}}^{R}\int_{z=0}^{\frac{E_{th}}{\eta \xi \tau P_{t}(\beta d_{K}^{-\alpha }+d_{1}^{-\alpha })}} (\frac{z^{k-1}}{(k-1)!})\\&\exp(\!\frac{-E_{th}d_{1}^{ \alpha }}{\eta \xi \tau P_{t}})f_{d_{1}}(r)f_{d_{K}}(r)f_{K}(k)\Big ]d_{z} d_{d_{K}} d_{d_{1}}+ \\
& \sum_{k = 2}^{\infty }\Big[\int_{d_{1}=0}^{R}\int_{d_{K}=d_{1}}^{R}
\int_{z=\frac{E_{th}}{\eta \xi \tau P_{t}(\beta d_{K}^{-\alpha }+d_{1}^{-\alpha })}}^{\infty}(\frac{z^{k-1}}{(k-1)!})\\&\exp(-\beta d_{1}^{\alpha } d_{K}^{-\alpha } z-z)  f_{d_{1}}(r)f_{d_{K}}(r)f_{K}(k)\Big ] d_{z} d_{d_{K}} d_{d_{1}}, 
\label{theoremUeq}
\end{aligned}
\end{equation}
where, $f_{K}(k) = \frac{e^{-\lambda \pi R^{2}}(\lambda \pi R^{2})^{k}}{k!}$ and $f_{d_{1}}(r)$ and $f_{d_{K}}(r)$ are obtained from \eqref{e10} and \eqref{e9}.
\end{theorem} 
\begin{proof}
A detailed proof is given in Appendix B.    
\end{proof}
Case C: $Pr_{max}<Pr$.\\
\begin{theorem}   \label{theorem3}
The lower bound of $\mu _{\phi }$ is given by 
\begin{equation}
\begin{aligned}
\!\mu_{\phi,L}\!\!=\!\!&\sum_{k=2}^{\infty}\int_{d_{1}=0}^{R}\int_{z=0}^{\frac{E_{th}d_{1}^{\alpha }}{\eta \xi \tau P_{t}(1+\beta )}}\frac{z^{k-1}}{(k-1)!}exp\left (- \left [\frac{E_{th}d_{1}^{\alpha }}{\eta \xi \tau P_{t}}  \right ] \right )\\&f_{d_{1}}(r)f_{K}(k)d_{z}d_{d_{1}} +\sum_{k=2}^{\infty }\int_{d_{1}=0}^{R}\int_{z=\frac{E_{th}d_{1}^{\alpha }}{\eta \xi \tau P_{t}(1+\beta )}}^{\infty }\frac{z^{k-1}}{(k-1)!}\\&exp(-(\beta +1)z)f_{d_{1}}(r)f_{K}(k)d_{z}d_{d_{1}}
\label{theorem1eqC}
\end{aligned}
\end{equation}
where $f_{K}(k) = \frac{e^{-\lambda \pi R^{2}}(\lambda \pi R^{2})^{k}}{k!}$ and $f_{d_{1}}(r)$ and $f_{d_{K}}(r)$ are obtained from \eqref{e10} and \eqref{e9}.
\end{theorem}
\begin{proof}
A detailed proof is given in Appendix C.    
\end{proof}

The upper bound of $\mu _{\phi }$, ($\mu _{\phi ,U}$) is same as Theorem \ref{theorem2}.  

\section{\ac{PAoI} Analysis}
The \ac{JSP} influences the \ac{PAoI} since it directly characterizes the probability of success or failure in a given time slot. Recall that for the non-preemptive scenario, newly arriving packets at the transmitter are dropped until the one in transmission is successfully delivered. Based on the derived framework for \ac{JSP}, the following Theorem characterizes an upper bound of the \ac{PAoI} of the system.
\begin{theorem}   \label{theorem4}
The upper bound of the \ac{PAoI} measured at the typical \ac{IoT} device is  
\begin{equation}
\begin{aligned}
&A_{i,NP}^{U}= Z_{a} + \frac{2}{\mu _{\phi,L}},
\label{e16}
\end{aligned}
\end{equation}
where lower bound of $\mu _{\phi }$ gives an upper bound of the \ac{PAoI}. 
\end{theorem}
\begin{proof}
A detailed proof is given in Appendix D.    
\end{proof}
\begin{theorem}   \label{theorem5}
The upper bound of the \ac{PAoI} measured at the typical \ac{IoT} device is  
\begin{equation}
\begin{aligned}
&A_{i,P}^{U}= Z_{a} + \frac{1}{\mu _{\phi,L}} + \frac{1}{q_{s}},
\label{e16p}
\end{aligned}
\end{equation}
where lower bound of $\mu _{\phi }$ gives an upper bound of the \ac{PAoI} and $q_{s} = \mu_{\phi ,L}+p_{a}\left ( 1-\mu _{\phi ,L} \right )$. 
\end{theorem}
\begin{proof}
A detailed proof is given in Appendix E.    
\end{proof}
We are interested in the upper bound of \ac{PAoI} since the \ac{PAoI} measures the maximum time elapsed since the last received packet at the destination was generated, just before a new packet is received at the destination. This gives a guarantee on the \ac{PAoI} of the system. In the next section, first, we validate and study the trends of the \ac{JSP} and then investigate the optimal slot partitioning factor to optimize either the \ac{JSP} or the \ac{PAoI}.

\section{Simulation Results}
In this section, we present simulation results to verify the derived expressions. Throughout this section, unless otherwise specified, we consider the system parameters as $E_{th} = 10$ mJ, $\alpha = 3$, $\sigma = 10$ bits, $\eta = 0.9$, $\xi = 0.4$, $\tau = 1$ sec, $\lambda = 0.003 $ m$^{-2}$,  $R = 60$ m, $B = 10$ kHz. 

\begin{figure}
\centering
\includegraphics[width=8.8cm, height=7cm]{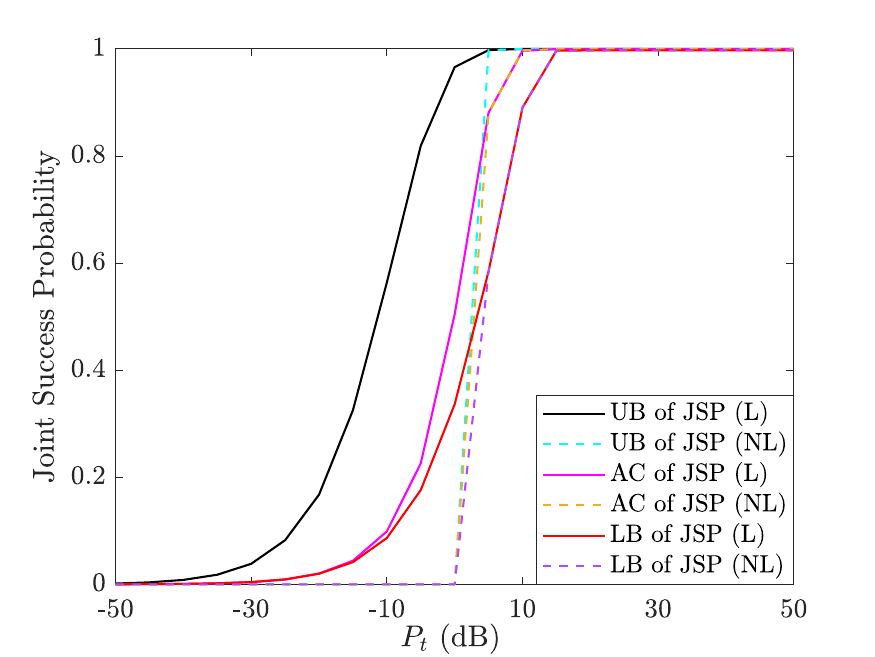}
\caption{The accuracy of the \ac{UB}, \ac{AC} and \ac{LB} of the \ac{JSP}, $\mu _{\phi }$ for both \ac{L EH} and \ac{NL EH} models with respect to the transmit power.}
\label{plot_1}   
\end{figure}

A wide range of bandwidths, from a few kHz to several GHz, can be utilized by IoT devices, contingent upon their applications and specifications. IoT devices with low bandwidth operate within the kHz range, making them ideal for low-power communication e.g., AM/FM radio that involves the transmission of minimal data in the range of a few Kbps using protocols like LoRaWAN and Zigbee. Medium-bandwidth IoT devices operate within the range of a few MHz, so moderate data transmission in the Mbps range can be used for real-time monitoring and control through protocols like Bluetooth and LTE-M.

According to the \ac{L EH} model as shown by the dashed lines in Fig. \ref{plot_1}, the actual value of \ac{JSP} matches the \ac{LB} for lower transmit power values and the \ac{UB} for higher transmit power values. On the contrary, we see from Fig.~\ref{plot_2} that for a varied range of deployment areas, the lower and the upper bounds closely follow the actual \ac{JSP}. The \ac{JSP} monotonically increases with the transmit power. This is because, with an increase in $P_t$, the \ac{EH} increases, while it has no impact on the \ac{SIR}, where $P_t$ appears both in the numerator and the denominator, and hence cancels out. Interestingly, with an increase of $R$ the \ac{JSP} first increases, reaches an optimal value and then decreases further. The initial increase is due to the enhancement in the \ac{EH} caused by an increase in the number of transmitters. Albeit, this deteriorates the \ac{SIR}, but for lower values of $R$, the \ac{SIR} already exceeds its corresponding threshold. On the contrary, for larger values of $R$, the \ac{SIR} deteriorates due to an increase in the number of interferers.

\ac{NL EH} as shown by dotted lines in Fig. \ref{plot_1} and Fig. \ref{plot_4NL} embodies a real-world scenario, as it simulates a realistic system where power constraints (both minimum and maximum) are imposed; specifically, no energy is harvested below a designated power threshold, and there exists a limit on the maximum \ac{EH}, which is essential in practical energy-harvesting systems.  

The plots of both figures, for the \ac{L EH} and \ac{NL EH} models, overlap at higher power levels, but separate at lower power levels. At low power levels, the plots diverge due to the \ac{NL EH} scenario, which establishes a minimum power threshold $P_{r,min}$ beneath which energy cannot be harvested. This results in a significant decline in \ac{JSP} and consequently the $\xi^{*}$ for maximizing the lower bound of \ac{JSP}, whereas in the \ac{L EH} scenario, the IoT continues to harvest energy linearly, yielding higher \ac{JSP}. At higher power levels, both scenarios yield similar \ac{JSP} values, as in the \ac{NL} case, \ac{EH} saturates at $P_{r,max}$, whereas in the \ac{L} case, \ac{EH} increases linearly with power and the \ac{EH}, both capped and uncapped are large enough to exceed the energy threshold, $E_{th}$ also the SIR condition is readily fulfilled. 

\begin{figure} 
\centering
\includegraphics[width=8.8cm, height=7cm]{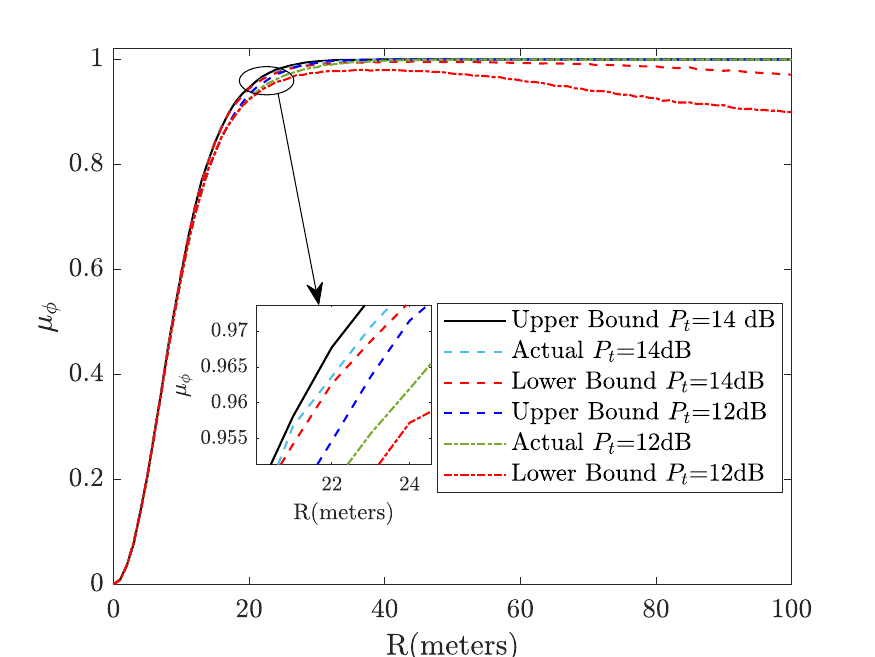}
\caption{Lower and upper bounds of joint success probability, $\mu _{\phi}$ with respect to the radius, R compared with the actual simulation when received power $P_{r}$ lies in between $P_{r,min}$ and $P_{r,max}$.}
\label{plot_2}
\end{figure}

\begin{figure}
\centering
\includegraphics[width=8.8cm, height=7cm]{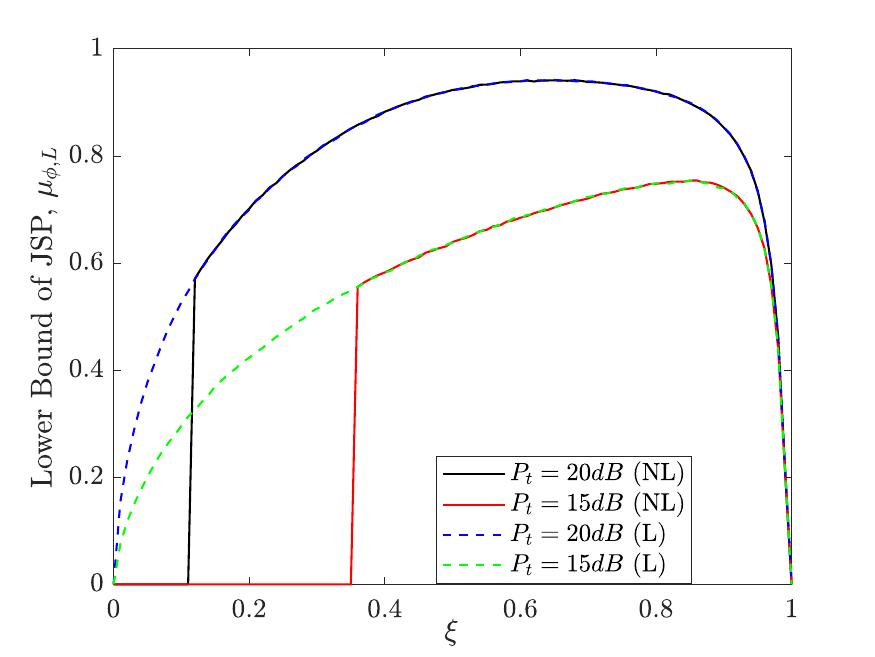}
\caption{Lower bound of joint success probability, $\mu _{\phi,L}$ with respect to fraction of time slot, $\xi$ for radius = 50 m.}
\label{plot_3}
\end{figure}

\ac{L EH} model as shown by dashed lines in Fig.~\ref{plot_3} and Fig.\ref{fig_16pre} confirm that even at small values of $\xi$, some energy is still harvested, albeit small. As $\xi$ increases, the \ac{EH} also increases, and the probability of meeting both the energy and \ac{SIR} thresholds improves. Therefore, the \ac{JSP} curve will likely increase smoothly while \ac{PAoI} will likely decrease smoothly since, there is no cap on the harvested energy. The Fig. \ref{fig_16pre} also illustrates that the \ac{PAoI} for the preemptive queue discipline is consistently lower than that of the non-preemptive queue discipline for all values of the time slot fraction, $\xi$. This outcome proves that the preemptive queue discipline outperforms the non-preemptive approach.
As $\xi$ further increases, despite higher energy harvesting, the overall \ac{JSP} declines because the system fails to meet the \ac{SIR} requirement due to insufficient transmission time. This trade-off between energy harvesting and data transmission time results in the \ac{JSP} curve declining as $\xi$  approaches 1.
In other words, when $\xi$ is small, the \ac{IoT} device does not harvest sufficient energy to be able to successfully decode the data packets in the data transmission phase, resulting in lower $\mu _{\phi,L}$, higher $A_{i,NP}^{U}$ and higher $A_{i,P}^{U}$ respectively. Similarly, when $\xi$ is large, the data transmission interval is low, resulting in a higher link rate requirement. This degrades $\mu _{\phi,L}$, $A_{i,NP}^{U}$ and $A_{i,P}^{U}$ respectively. This establishes the existence of an optimal slot partitioning factor $\xi^{*}$ that maximizes JSP and minimizes PAoI, for a disc of a given radius, corresponding to each value of transmit power. This is clearly demonstrated in Fig. \ref{plot_4} for discs of varying radii, where we can also  observe that as the power transmitted increases, \ac{EH} also increases; thus, to further increase the maximum possible \ac{JSP}, greater time can be allotted to information decoding, thereby decreasing $\xi^{*}$. In conclusion,  $\xi^{*}$ for lower values of transmit power are always greater than $\xi^{*}$ values for higher values of transmit power.

In \ac{NL EH} model, as shown by dotted lines in Fig.~\ref{plot_3} and Fig.\ref{fig_16pre}. At low $\xi$, if the received power $P_{r}$ is below the $P_{r,min}$ threshold, then no energy is harvested. This can result in a sharper decline or a flat region in the \ac{JSP} and \ac{PAoI} plot, due to the energy constraint not being met. Once $P_{r}$ exceeds $P_{r,min}$, energy harvesting starts and \ac{JSP} will increase and \ac{PAoI} will decrease, but at a slower rate due to the conditions and caps imposed on $P_{r}$. When the received power exceeds $P_{r,max}$, the energy is capped, meaning that even with a further increase in $\xi$, the energy harvesting no longer increases linearly. This results in the \ac{JSP} and \ac{PAoI} curve flattening at high $\xi$ values. As $\xi$ increases further, more time is spent harvesting energy and less time is available for data transmission. This reduction in available time for transmitting data negatively impacts the system's ability to meet the \ac{SIR} requirement. Specifically, the threshold for the \ac{SIR}, represented by $\beta$ increases as $\xi$ increases because the available time for transmission becomes smaller. As $\xi$ approaches 1, the transmission time (1 - $\xi$) becomes close to zero, making it harder to meet the \ac{SIR} threshold, and thus reducing the \ac{JSP} and increase \ac{PAoI}. While more energy is harvested at high $\xi$, the benefit of increased energy is offset by the decreasing transmission time and the increasing SIR threshold $\beta$. Essentially, while the system might have sufficient energy, the data transmission reliability deteriorates, and the likelihood of successful joint energy harvesting and data transmission declines.

From \eqref{e16}, we see that the term $Z_{a}=\left ( \frac{1}{p_{a}}-1\right ) $ is independent of $\xi$ so $A_{i,NP}^{U}$ is minimized and $\mu _{\phi ,L}$ is maximized for the same optimal $\xi$. Upon differentiating \eqref{e16p} i.e. $A_{i,P}^{U}$ with respect to $\mu _{\phi ,L}$, we get,
$ \frac{d A_{i,P}^{U}}{d \mu _{\phi ,L}}=-\frac{1}{\mu _{\phi ,L}^{2}}-\frac{(1-p_{a})}{(\mu _{\phi ,L}+p_{a}(1-\mu _{\phi ,L} )^{2})^{2}}<0$ for the range $p_{a}\in[0,1]$. This  means that $A_{i,P}^{U}$ is monotonically decreasing with respect to $\mu _{\phi ,L}$ so $A_{i,P}^{U}$ is minimized simultaneously when $\mu _{\phi ,L}$ is maximized for the same optimal $\xi$. Thus we conclude that $\mu _{\phi ,L}$ is maximized while $A_{i,P}^{U}$, $A_{i,NP}^{U}$ are minimized for the same value of $\xi^{*}$ which is represented in Fig. \ref{plot_4}. Here we can also observe that the optimal $\xi$ plots for maximizing \ac{JSP} are the same as the optimal $\xi$ plots for minimizing \ac{PAoI} where the actual case plots lie in between the upper and lower bound plots. Another observation here is that as the radius increases, the plots of the bounds get farther away from each other and also from the actual case in a peculiar way that encompasses the crossing of the $R = 50$ m and $R = 200$ m plots of $\xi^{*}$ for maximizing the lower bound of JSP and minimizing the upper bound of PAoI. This is explained by Fig.~\ref{fig_14}, \ac{L EH} model as follows.
\begin{figure}
\centering
\includegraphics[width=8.8cm, height=7cm]{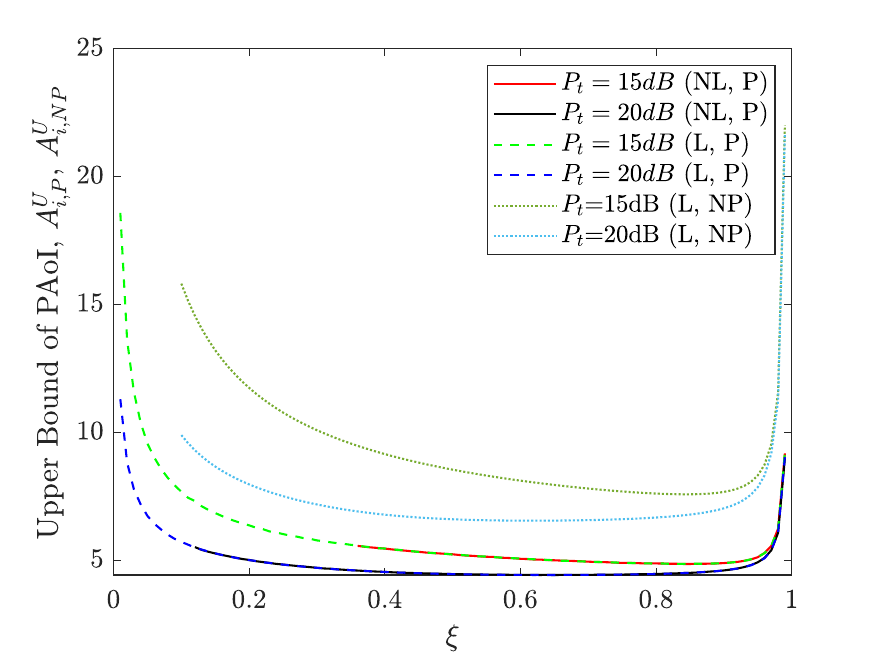}
\caption{Upper Bound of \ac{PAoI} under both preemptive and non-preemptive queueing disciplines, $A_{i,P}^{U}$, $A_{i,NP}^{U}$ with respect to fraction of time slot, $\xi$.}
\label{fig_16pre} 
\end{figure}

\begin{figure}
\centering
\includegraphics[width=8.8cm, height=7cm]{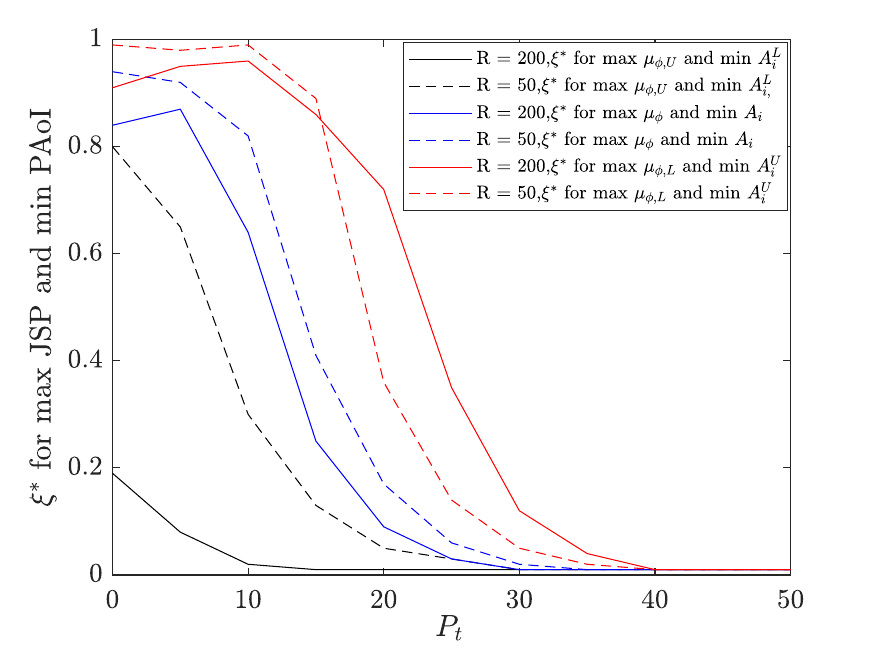}
\caption{Optimal slot partitioning, $\xi^{*}$ for maximizing upper $\mu_{\phi,U}$, actual $\mu_{\phi}$ and lower $\mu_{\phi,L}$ bounds of joint success probability and minimizing lower $A_{i,P}^{L}$, $A_{i,NP}^{L}$, actual $A_{i,P}$, $A_{i,NP}$ and upper $A_{i,P}^{U}$, $A_{i,NP}^{U}$ bounds of \ac{PAoI}, with respect to transmit power, $P_{t}$ in \ac{L EH} model.}
\label{plot_4}
\end{figure}

\begin{figure}
\centering
\includegraphics[width=8.8cm, height=7cm]{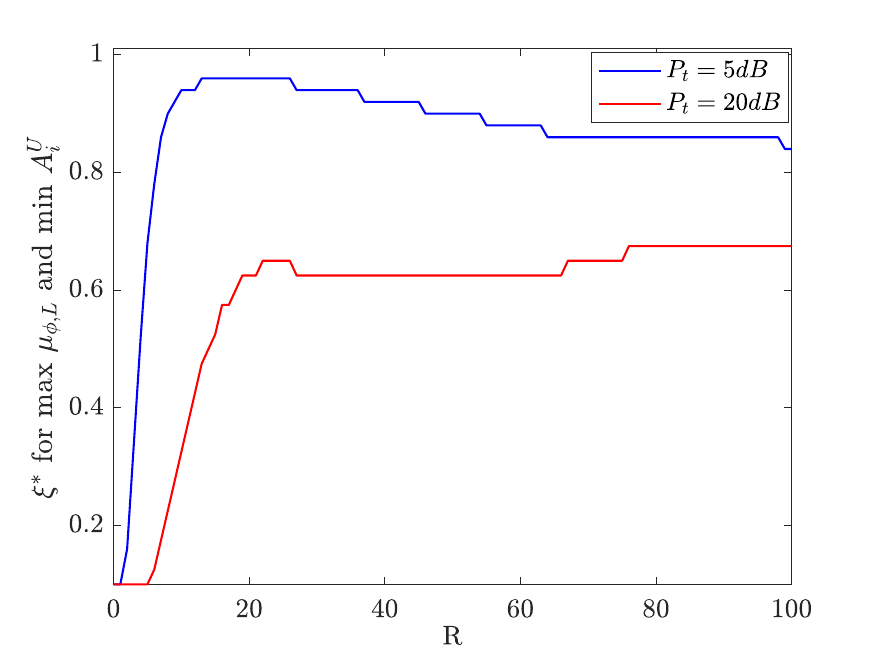}
\caption{Optimal slot partitioning, $\xi^{*}$ for maximizing lower bound of JSP, $\mu_{\phi ,L}$ and minimizing upper bounds of PAoI, $A_{i,P}^{U}$, $A_{i,NP}^{U}$ with respect to radius, R.}
\label{fig_14}
\end{figure}

\begin{figure}
\centering
\includegraphics[width=8.8cm, height=7cm]{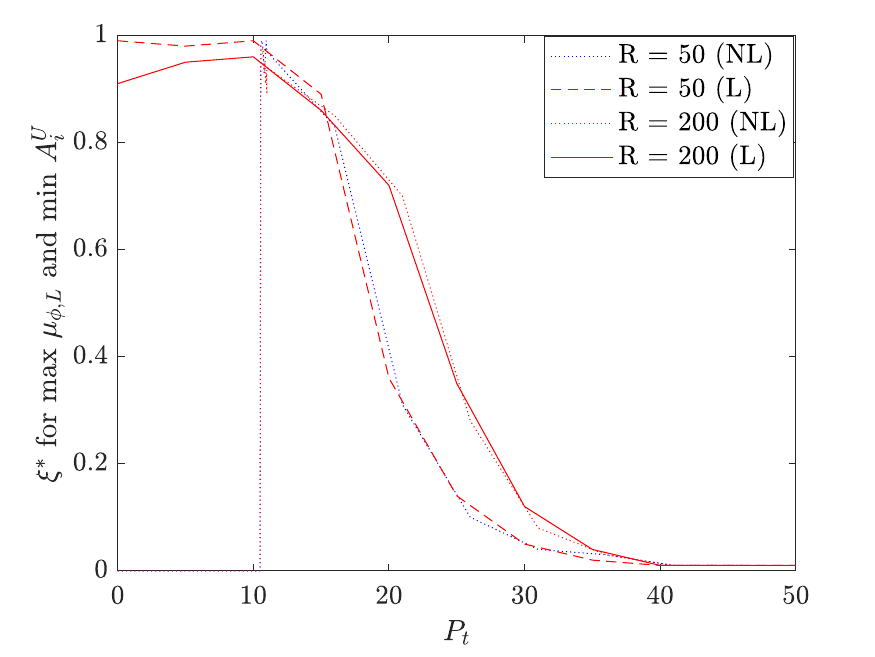}
\caption{Optimal slot partitioning, $\xi^{*}$ for maximizing lower bound of joint success probability, $\mu_{\phi,L}$ and minimizing upper bound of \ac{PAoI}, $A_{i,P}^{U}$, $A_{i,NP}^{U}$ with respect to transmit power, $P_{t}$ in both \ac{L EH} and \ac{NL EH} model.}
\label{plot_4NL}
\end{figure}

Fig.~\ref{fig_14}, \ac{L EH} model shows that when the transmit power is low, say $P_{t}$ = 5 dB and the radius is very small, the number of interfering transmitters is also very less and barely grows with a small increase in the radius. At the same time, the signaling transmitter is very close to the receiver, so the $\ac{SIR}>\beta$ threshold condition is easily met and the IoT device also harvests a small amount of energy, enough to meet the threshold, $\ac{EH}>E_{th}$. However, with a small increase in the radius, the amount of energy harvested drops, therefore, \ac{EH} time has to be increased by increasing $\xi^{*}$ which maximizes \ac{JSP} and minimizes \ac{PAoI}.

At higher R values, as R increases, the number of interference transmitters increases significantly. This means that excessive interference and insufficient signal power collectively hinder signal decoding, which in turn lowers the \ac{SIR}. Thus, in order to maximize \ac{JSP} and minimize \ac{PAoI}, more time must be allocated for signal decoding, by decreasing $\xi^{*}$.

When the transmit power is high, say $P_{t}$ = 20 dB and the radius is very small then with a small increase in the radius, $\xi^{*}$ is increased for the same reason as previously stated; however, for larger R values, as R increases, the number of interferers increase significantly. Since stronger transmit power overcomes interference more efficiently, so the decrease in the \ac{SIR} is not enough to hamper signal decoding, hence, more energy can be harvested by increasing $\xi^{*}$ to further maximize \ac{JSP}. This ensures that even with lesser time available for data transmission, the transmissions are successful a lot more often.
This efficient use of higher transmit power allows \ac{JSP} for higher values of radius to be greater than \ac{JSP} for lower values of radius.

A really interesting observation in Fig. \ref{plot_4} in the \ac{L EH} model in the actual case plot drawn for different R values is that, as the radius of the disc increases, the number of interferers increase significantly, thus decreasing \ac{SIR} so the signal decoding time needs to be increased by decreasing the optimum charging phase duration, $\xi^{*}$ in order to maximize \ac{JSP} for all values of transmit power, $P_{t}$. 

One clear distinction between the plots of the upper bound and the actual case of $\xi^{*}$ for minimizing PAoI in Fig. \ref{plot_4} is that, in the upper bound plot, at lower values of transmit power, $P_{t}$ and in the actual case plot, at all values of power transmitted by the signaling transmitter,   as the radius increases, the number of interferers increase significantly, thereby increasing interference. As a result, signal decoding time at IoT must be increased by decreasing the $\xi^{*}$.

On the contrary, we can see in the plot of the upper bound of $\xi^{*}$ for minimizing PAoI that, at higher values of transmit power, as the radius increases, so does the $\xi^{*}$. This occurs in the boundary case because higher transmit power values overcome interference more efficiently. As a result, IoT can leverage this situation to maximize JSP by increasing $\xi^{*}$. By increasing the $\xi^{*}$, IoT can harvest more energy, allowing signal decoding to improve even as signal decoding time at IoT decreases. 
\vspace{-0.3cm}
\section{Conclusion}
This paper considered an \ac{IoT} network where the \ac{IoT} devices harvest energy from the power received from the transmitters and use it to decode the data packets. We then define and derive the lower and upper bounds for the metric \ac{JSP} and the upper bound for the metric \ac{PAoI}.
Our analytical results show how time slot partitioning factor, deployment area, and other wireless network parameters affect \ac{JSP} and \ac{PAoI}. The key observations were that in NL EH, no energy is harvested below a certain power threshold, and the highest harvested energy is restricted, which is essential for the design of practical energy-harvesting systems. The preemptive queue discipline consistently outperforms the non-preemptive approach, as evidenced by the fact that the \ac{PAoI} is lower for the preemptive queue discipline than for the non-preemptive queue discipline for all values of the time slot fraction, $\xi$. An optimal slot partitioning factor $\xi^{*}$ exists that maximizes JSP and minimizes PAoI in both preemptive and non-preemptive queuing disciplines for a disc of a specific radius, corresponding to each value of transmit power. Further, we illustrate how this factor stays the same across each of the three cases. As the radius increases, the number of interferers increase, leading to increased interference, thus reducing $\xi^{*}$ in the upper bound plot of $\xi^{*}$ for minimizing PAoI, for only the lower values of transmit power, $P_{t}$ and for all values of transmit power, $P_{t}$ in the actual case plot. 
However, at higher values of transmit power, $\xi^{*}$ grows with radius in the upper bound plot of $\xi^{*}$ for minimizing PAoI, because higher transmit power values help the IoT to overcome interference more efficiently, thus maximizing JSP by increasing $\xi^{*}$ despite facing a decrease in decoding time. As a promising avenue of future work, one can analyze protocols like ALOHA and CSMA-CD with a channel access parameter which can be optimized using some suitable RL algorithm.
\appendices
\section{Proof of Theorem 1} \label{AppendixA} 
\vspace{-0.25cm}
The lower bound of $\mu _{\phi }$ is given as follows: $\mu _{\phi ,L}=\mathbb{P}[(E_{1}^{j,min}> E_{th}) 
, (\Upsilon_{1}^{j,min}> \beta )]$. Using \eqref{e11_lower} and \eqref{SII_min}, we can write $\mu _{\phi ,L}$ as 
\begin{equation}
\begin{aligned} 
\mu _{\phi ,L}\!&=\mathbb{P}\bigg(g_{1}^{j}>\underset{\Theta_{1}}{\underbrace{ d_{1}^{\alpha}\bigg[ \frac{E_{th}}{\eta\xi\tau P_{t}}-d_{K}^{-\alpha} \sum_{k=2}^{K}g_{k}^{j}}\bigg]}, g_{1}^{j}>\underset{\Theta_{2}}{\underbrace{\beta\sum_{k=2}^{K}g_{k}^{j}}} \bigg) \nonumber
\end{aligned}
\end{equation}
\begin{equation}
\begin{aligned} 
&= \mathbb{P} \big (g_{1}^{j}>\max(\Theta_{1} , \Theta_{2}) \big)\\
&= \underset{R_{1}}{\underbrace{\mathbb{P} \big(g_{1}^{j}>\Theta_{1} , \Theta_{1}>\Theta_{2}\big)}} + \underset{R_{2}}{\underbrace{\mathbb{P}\big(g_{1}^{j}>\Theta_{2} , \Theta_{2}>\Theta_{1}\big)}} 
\label{e15}
\end{aligned}
\end{equation}
In the above expression, there are two probabilities that need to be solved. We solve them one by one. If $\Theta_{1}>\Theta_{2}$, then we have $\big[E_{th}/\eta\xi\tau P_{t} - d_{K}^{-\alpha}Z \big]
>\beta d_{1}^{-\alpha} Z$, where $Z=\sum_{k=2}^{K}g_{k}^{j}$. To solve $R_{1}$, we have
\newline
\begin{equation}
\begin{aligned}   \label{R1}
R_{1}\!=&\mathbb{P}\Bigg(\!g_{1}^{j}\!>\!d_{1}^{\alpha}\!\bigg[\!\frac{E_{th}}{\eta\xi\tau P_{t}}\!-\!Z d_{K}^{-\alpha}\bigg]\!,\!
&\!\!\!\!Z\!\!<\!\frac{E_{th}}{\eta \xi \tau P_{t}\big[ \beta d_{1}^{-\alpha}\!+\! d_{K}^{-\alpha}\!\big]}\!\!\Bigg)
\end{aligned}
\end{equation} 
Let $t=d_{1}^{\alpha}\bigg[\frac{E_{th}}{\eta\xi\tau P_{t}}-Z d_{K}^{-\alpha}\bigg]$. If $t\geq 0 \Rightarrow Z\leq \frac{E_{th}d_{K}^{\alpha}}{\eta\xi\tau P_{t}}$. If $t<0 \Rightarrow Z>\frac{E_{th}d_{K}^{\alpha}}{\eta\xi\tau P_{t}}$. The condition $t<0$ will not occur since channel fading cannot be negative. Hence, we have 
\begin{equation}
\begin{aligned} 
R_{1}&= \mathbb{P}\bigg(g_{1}^{j}>t,t> 0 ,
Z<\frac{E_{th}}{\eta \xi \tau P_{t}\big[ \beta d_{1}^{-\alpha} + d_{K}^{-\alpha}\big]} \bigg)
\end{aligned}
\end{equation} 
On substituting $t=d_{1}^{\alpha}\bigg[\frac{E_{th}}{\eta\xi\tau P_{t}}-Zd_{K}^{-\alpha}\bigg]$, we have
\begin{equation}
\begin{aligned} 
R_{1}&= \mathbb{P}\bigg(g_{1}^{j}>d_{1}^{\alpha}\bigg[\frac{E_{th}}{\eta\xi\tau P_{t}}-Zd_{K}^{-\alpha}\bigg],\\
&Z< \min \bigg(\frac{E_{th}}{\eta \xi \tau P_{t}\big[ \beta   d_{1}^{-\alpha} + d_{K}^{-\alpha}\big]}, \frac{E_{th}d_{K}^{\alpha}}{\eta\xi\tau P_{t}} \bigg) \bigg)
\end{aligned}
\end{equation} 
Since,$\frac{E_{th}}{\eta \xi \tau P_{t}\!\big[\! \beta  d_{1}^{-\alpha}\!\!+ d_{K}^{-\alpha}\!\big]\!}$is always less than$\frac{E_{th}d_{K}^{\alpha}}{\eta\xi\tau P_{t}}$.We can write 
\begin{equation}
\begin{aligned} 
R_{\!1}&\!=\!\mathbb{P}\bigg(\!g_{1}^{j}\!>\!d_{1}^{\alpha}\!\bigg[\!\frac{E_{th}}{\eta\xi\tau P_{t}}\!-\!Zd_{K}^{-\alpha}\!\bigg]\!,\!
Z\!<\!\frac{E_{th}}{\eta \xi \tau P_{t}\big[ \beta   d_{1}^{-\alpha}\!+\!d_{K}^{-\alpha}\!\big]\!}\!\bigg)\\ 
&\!=\!\mathbb{E}_{Z}\!\bigg\{\!\!\exp\!\bigg(\!\!\!-\!d_{1}^{\alpha}\!\bigg[\!\frac{E_{th}}{\eta\xi\tau P_{t}}\!-\!Zd_{K}^{-\alpha}\!\bigg]\!,
Z\!\!<\!\!\frac{E_{th}}{\eta \xi \tau P_{t}\!\big[\!\beta   d_{1}^{-\alpha}\!+\!d_{K}^{-\alpha}\!\big]}\!\bigg)\!\!\bigg\}\\
&=\!\!\int_{0}^{\frac{E_{th}}{\eta \xi \tau P_{t}\!\big[\! \beta d_{1}^{-\alpha}\!+d_{K}^{-\alpha}\!\big]\!}}\exp\bigg(\!\!\!-\!d_{1}^{\alpha}\!\bigg[\!\frac{E_{th}}{\eta\xi\tau P_{t}}\!-\!Zd_{K}^{-\alpha}\bigg]\!\bigg)f_{Z}(Z)dZ
\label{R1appendix}
\end{aligned}
\end{equation} 
The PDF of $Z$, $f_{Z}(Z)$ is as follows: $f_{Z}(Z)=\frac{e^{-Z}Z^{k-1}}{(k-1)!}$. Using $f_{Z}(Z)$ in\eqref{R1appendix}, we get 
\begin{equation}
\begin{aligned}
R_{\!1}\!\!&=\!\!\int_{0}^{\frac{E_{th}}{\eta \xi \tau P_{t}\!\big[\! \beta   d_{1}^{-\alpha}\!\!+ d_{K}^{-\alpha}\big]\!\!}}\!\!\frac{Z^{k-1}}{(k-1)!}\! \nonumber
\end{aligned}
\end{equation}
\begin{equation}
\begin{aligned}
&\exp\!\bigg(\!\!\!-\!\!\bigg[\frac{E_{th}d_{1}^{\alpha}}{\eta\xi\tau P_{t}}\!-\!Zd_{1}^{\alpha}d_{K}^{-\alpha}+\!Z\bigg]\bigg)dZ
\label{e20}
\end{aligned}
\end{equation}
To Solve $R_{2}$, we can write
\begin{equation}
\begin{aligned} 
R_{2}&=\mathbb{P}\bigg(g_{1}^{j}>\beta Z, 
Z<\frac{E_{th}}{\eta \xi \tau P_{t}\big[ \beta   d_{1}^{-\alpha} + d_{K}^{-\alpha}\big]}\bigg)\\
&= \mathbb{E}_{Z}\bigg\{\exp\big ( -\beta Z \big ),Z> \frac{E_{th}}{\eta \xi \tau P_{t}\big ( \beta d_{1}^{-\alpha }+d_{K}^{-\alpha} \big)} \bigg\}\\
&= \int_{\frac{E_{th}}{\eta \xi \tau P_{t}\big ( \beta d_{1}^{-\alpha }+d_{K}^{-\alpha} \big)}}^{\infty } \exp\big (-\beta Z  \big )f_{Z}(z)dZ
\label{R2appendix}
\end{aligned}
\end{equation}
The PDF of $Z$, $f_{Z}(Z)$ is as follows: $f_{Z}(Z)=\frac{e^{-Z}Z^{k-1}}{(k-1)!}$. 
Using $f_{Z}(Z)$ in \eqref{R2appendix}, we get
\begin{equation}
\begin{aligned}
R_{2}&= \int_{\frac{E_{th}}{\eta \xi \tau P_{t}\big ( \beta d_{1}^{-\alpha }+d_{K}^{-\alpha} \big)}}^{\infty } \exp\big (-\beta Z -Z \big )\frac{Z^{k-1}}{(k-1)!}dZ
\label{e22}
\end{aligned}
\end{equation}
Replacing \eqref{e20} and \eqref{e22} in \eqref{e15} , we get \eqref{theorem1eq}.
\section{Proof of Theorem 2} \label{AppendixB}
\newpage
Upper Bound of $\mu _{\phi }$ is given as follows: $\mu _{\phi ,U} = \mathbb{P}[(E_{1}^{j,max}> E_{th}),(\Upsilon_{1}^{j,max}> \beta )]$. Using \eqref{e11_upper} and \eqref{SII_max}, we can write $\mu _{\phi}^{U}$ as
\begin{equation}
\begin{aligned} 
\mu _{\phi}^{U}\!&=\!\mathbb{P}\!\bigg(\!g_{1}^{j}\!>\!\underset{\Theta_{3}}{\underbrace{ d_{1}^{\alpha} \!  \bigg[ \!\frac{E_{th}}{\eta\xi\tau P_{t}}\!-\! d_{1}^{-\alpha} \!\sum_{k=2}^{K}g_{k}^{j}}\!   \bigg]  }\!  ,\!
g_{1}^{j}\!>\!\underset{\Theta_{4}}{\underbrace{\beta d_{1}^{\alpha} d_{K}^{-\alpha}\! \sum_{k=2}^{K}\!g_{k}^{j}}}\!\bigg)\!  
\nonumber
\end{aligned}
\end{equation}
\begin{equation}
\begin{aligned} 
&=\mathbb{P}(g_{1}^{j}>\max(\Theta_{3} , \Theta_{4}))\\
&=\underset{R_{3}}{\underbrace{\mathbb{P}(g_{1}^{j}>\Theta_{3} , \Theta_{3}>\Theta_{4})}} + \underset{R_{4}}{\underbrace{\mathbb{P}(g_{1}^{j}>\Theta_{4} , \Theta_{4}>\Theta_{3})}}.
\label{e24}
\end{aligned}
\end{equation}
In the above expression, there are two probabilities that need to be solved. We solve them one by one. If $\Theta_{3}>\Theta_{4}$ then we have $\bigg[ {E_{th}}/{\eta\xi\tau P_{t}}- d_{1}^{-\alpha} Z   \bigg]
>\beta  d_{K}^{-\alpha}Z$ ,where $Z=\sum_{k=2}^{K}g_{k}^{j}$. \\ \\
To solve $R_{3}$, we have
\begin{equation}
\begin{aligned} 
R_{3}\!\!=&\mathbb{P}(g_{1}^{j}\!>\!d_{1}^{\alpha}\!\bigg[\!\frac{E_{th}}{\eta\xi\tau P_{t}}\!-\!Zd_{1}^{-\alpha}\!\bigg]\nonumber , 
Z\!<\frac{E_{th}}{\eta \xi \tau P_{t}\big[ \beta     d_{K}^{-\alpha} + d_{1}^{-\alpha}\big]})\\
\end{aligned}
\end{equation}
Let $t = d_{1}^{\alpha}\bigg[\frac{E_{th}}{\eta\xi\tau P_{t}}-Zd_{1}^{-\alpha}\bigg]$. If  $t>0$
$\Rightarrow Z<\frac{E_{th}d_{1}^{\alpha}}{\eta\xi\tau P_{t}} $
If  $t<0$
$\Rightarrow Z>\frac{E_{th}d_{1}^{\alpha}}{\eta\xi\tau P_{t}} $.
The condition $t < 0$ will not occur
since channel fading cannot be negative. Hence, we have
\begin{equation}
\begin{aligned} 
R_{3}=&\mathbb{P}\bigg(g_{1}^{j}>t,t> 0 ,
Z<\frac{E_{th}}{\eta \xi \tau P_{t}\big[ \beta     d_{K}^{-\alpha} + d_{1}^{-\alpha}\big]} \bigg)
\end{aligned}
\end{equation}
On substituting $t = d_{1}^{\alpha}\bigg[\frac{E_{th}}{\eta\xi\tau P_{t}}-Zd_{1}^{-\alpha}\bigg]$, we have
\begin{equation}
\begin{aligned} 
R_{3}=&\mathbb{P}\bigg(g_{1}^{j}>d_{1}^{\alpha}\bigg[\frac{E_{th}}{\eta\xi\tau P_{t}}-Zd_{1}^{-\alpha}\bigg]\nonumber ,\\
&Z< min \bigg(\frac{E_{th}}{\eta \xi \tau P_{t}\big[ \beta     d_{K}^{-\alpha} + d_{1}^{-\alpha}\big]}, \frac{E_{th}d_{1}^{\alpha}}{\eta\xi\tau P_{t}} \bigg) \bigg).
\end{aligned}
\end{equation}
Since $\frac{E_{th}}{\eta \xi \tau P_{t}\big[ \beta     d_{K}^{-\alpha} + d_{1}^{-\alpha}\big]}$ is always less than $\frac{E_{th}d_{1}^{\alpha}}{\eta\xi\tau P_{t}}$. We can write
\begin{equation}
\begin{aligned}
R_{3}\!=& \mathbb{P}\bigg(\!g_{1}^{j}\!>\!d_{1}^{\alpha}\!\bigg[\frac{E_{th}}{\eta\xi\tau P_{t}}\!-\!Zd_{1}^{-\alpha}\!\bigg]\nonumber ,
Z\!<\! \frac{E_{th}}{\eta \xi \tau P_{t}\big[ \beta     d_{K}^{-\alpha}\! +\! d_{1}^{-\alpha}\big]}\! \bigg)\\ 
=&\!\mathbb{E}_{Z}\bigg\{\!\! \exp \!\bigg(\!\!\!-\!d_{1}^{\alpha}\!\bigg[\!\frac{E_{th}}{\eta\xi\tau P_{t}}\!-\!Zd_{1}^{-\alpha}\!\bigg]\!\nonumber ,
Z\!\!<\! \!\frac{E_{th}}{\eta \xi \tau P_{t}\!\big[ \beta     d_{K}^{-\alpha}\!\!\! +\! d_{1}^{-\alpha}\big]}\! \bigg)\!\! \bigg\}
\nonumber
\end{aligned}
\end{equation}
\begin{equation}
\begin{aligned} 
=\!\!&\int_{0}^{\frac{E_{th}}{\eta \xi \tau P_{t}\big[ \beta     d_{K}^{-\alpha} \!+ d_{1}^{-\alpha}\!\big]}}\!\!\exp\!\bigg(\!\!\!-\!d_{1}^{\alpha}\bigg[\!\frac{E_{th}}{\eta\xi\tau P_{t}}\!-\!Zd_{1}^{-\alpha}\!\bigg]\!\bigg)f_{Z}(Z)dZ.
\label{R1appendixB}
\end{aligned}
\end{equation}
The PDF of $Z$, $f_{Z}(Z)$ is as follows: $f_{Z}(Z)=\frac{e^{-Z}Z^{k-1}}{(k-1)!}$. Using $f_{Z}(Z)$ in \eqref{R1appendixB}, we get
\begin{equation}
\begin{aligned} 
R_{3}\!=&\int_{0}^{\!\frac{E_{th}}{\eta \xi \tau P_{t}\big[ \beta     d_{K}^{-\alpha} + d_{1}^{-\alpha}\big]}}\!
&\!\!\!\!\!\exp\bigg(\!\!-\!\bigg[\frac{E_{th}d_{1}^{\alpha}}{\eta\xi\tau P_{t}}\bigg]\bigg)\!\frac{Z^{k-1}}{(k-1)!}
d{Z}.
\label{e27}
\end{aligned}
\end{equation}
To Solve $R_{4}$, we can write
\begin{equation}
\begin{aligned} 
R_{4}=&\mathbb{P} \bigg(g_{1}^{j}>\beta d_{1}^{\alpha }d_{K}^{-\alpha }Z, 
Z<\frac{E_{th}}{\eta \xi \tau P_{t}\big[ \beta     d_{K}^{-\alpha} + d_{1}^{-\alpha}\big]}\bigg)
\nonumber
\end{aligned}
\end{equation}
\begin{equation}
\begin{aligned} 
=& \mathbb{E}_{Z}\bigg\{\exp\big ( -\beta d_{1}^{\alpha }d_{K}^{-\alpha } Z \big ),Z> \frac{E_{th}}{\eta \xi \tau P_{t}\big ( \beta d_{K}^{-\alpha }+d_{1}^{-\alpha} \big)} \bigg\}\\
=& \int_{\frac{E_{th}}{\eta \xi \tau P_{t}\big ( \beta d_{K}^{-\alpha }+d_{1}^{-\alpha} \big)}}^{\infty } \exp\big (-\beta d_{1}^{\alpha }d_{K}^{-\alpha }Z  \big )f_{Z}(z)dZ.
\label{R2appendixB}
\end{aligned}
\end{equation}
The PDF of $Z$, $f_{Z}(Z)$ is as follows: $f_{Z}(Z)=\frac{e^{-Z}Z^{k-1}}{(k-1)!}$. \\ Using $f_{Z}(Z)$ in \eqref{R2appendixB}, we get
\begin{equation}
\begin{aligned}
R_{4}\!=& \!\!\int_{\frac{E_{th}}{\eta \xi \tau P_{t}\big (\! \beta d_{K}^{-\alpha }\!+d_{1}^{-\alpha} \big)}}^{\infty }\!\!\!\! \exp\big (-\beta d_{1}^{\alpha }d_{K}^{-\alpha } Z -Z \big )\frac{Z^{k-1}}{(k-1)!}dZ.
\label{e29}
\end{aligned}
\end{equation}
Replacing \eqref{e27} and \eqref{e29} in \eqref{e24}, we get,
\eqref{theoremUeq}
\vspace{-0.1in}
\vspace{0.404cm}

\section{Proof of Theorem 3} \label{AppendixC}
The lower bound of $\mu _{\phi }$ is given as follows: $\mu _{\phi ,L}=\mathbb{P}[(E_{1}^{j}> E_{th}) 
, (\Upsilon_{1}^{j,min}> \beta )]$. Using \eqref{e2_lower} and \eqref{SII_min}, we can write $\mu _{\phi ,L}$ as 
\begin{equation}
\begin{aligned} 
\mu _{\phi ,L}&= P\left [ g_{1}^{j}>\underset{\Theta_{5}}{\underbrace{d_{1}^{\alpha }\left ( \frac{E_{th}}{\eta \xi \tau P_{t}}-d_{1}^{-\alpha }\sum_{k=2}^{K}g_{k}^{j} \right )}} ,g_{1}^{j}>\underset{\Theta _{6}}{\underbrace{\beta \sum_{k=2}^{K}g_{k}^{j}}}\right ]\\&
= P\left [ g_{1}^{j}>max(\Theta_{5},\Theta_{6}) \right ]\nonumber
\end{aligned} 
\end{equation}
\begin{equation}
\begin{aligned} 
&= \underset{R_{5}}{\underbrace{\mathbb{P}\left ( g_{1}^{j}>\Theta_{5},\Theta_{5}>\Theta_{6} \right )}}+\underset{R_{6}}{\underbrace{\mathbb{P}\left ( g_{1}^{j}>\Theta_{6},\Theta_{6}>\Theta_{5} \right )}}
\label{eq33}
\end{aligned} 
\end{equation}
In the above expression, there are two probabilities that need to be solved. We solve them one by one. If $\Theta_{5}>\Theta_{6}$, then we have $d_{1}^{\alpha }\left ( \frac{E_{th}}{\eta \xi \tau P_{t}}-d_{1}^{-\alpha }Z \right )>\beta Z$ where $Z= \sum_{k=2}^{K}g_{k}^{j}$.
To solve $R_{5}$, we have\\
\begin{equation}
\begin{aligned}
R_{5}\!=\!\mathbb{P}\!\left ( \!g_{1}^{j}\!>\!d_{1}^{\alpha }\!\left ( \!\frac{E_{th}}{\eta \xi \tau P_{t}}-d_{1}^{-\alpha }Z\! \right ), Z\!<\!\frac{E_{th}d_{1}^{\alpha }}{\eta \xi \tau P_{t}(1+\beta )}\! \right )
\end{aligned}
\end{equation}
Let $t = d_{1}^{\alpha }\left ( \frac{E_{th}}{\eta \xi \tau P_{t}}-d_{1}^{-\alpha }Z \right )$. If $t \geq 0$ $\Rightarrow Z<\frac{E_{th}d_{1}^{\alpha }}{\eta \xi \tau P_{t}}$. If $t < 0$ $\Rightarrow Z>\frac{E_{th}d_{1}^{\alpha }}{\eta \xi \tau P_{t}}$. The condition $t<0$ will not occur since channel fading cannot be negative. Hence we have
\begin{equation}
\begin{aligned}
R_{5} = \mathbb{P}\left ( g_{1}^{j}>t,t>0,Z<\frac{E_{th}d_{1}^{\alpha }}{\eta \xi \tau P_{t}(1+\beta )} \right ) 
\end{aligned}
\end{equation}
On substituting $t = d_{1}^{\alpha }\left ( \frac{E_{th}}{\eta \xi \tau P_{t}}-d_{1}^{-\alpha }Z \right )$, we have 
\begin{equation}
\begin{aligned}
R_{5}\!\!&=\!\!\mathbb{P}\!\left (\!\! g_{1}^{j}\!\!>\!\!d_{1}^{\alpha }\!\!\left [\! \frac{E_{th}}{\eta \xi \tau P_{t}}\!-\!d_{1}^{-\alpha }Z \!\right ]\!\!,Z\!\!<\!\!min\!\!\left ( \!\frac{E_{th}d_{1}^{\alpha }}{\eta \xi \tau P_{t}}\!,\!\frac{E_{th}d_{1}^{\alpha }}{\eta \xi \tau P_{t}(1\!\!+\!\!\beta )}\!\! \right )\!\! \right )\\&
=\mathbb{P}\left ( g_{1}^{j}\!\!>\!\!d_{1}^{\alpha }\!\!\left [\! \frac{E_{th}}{\eta \xi \tau P_{t}}\!-\!d_{1}^{-\alpha }Z \!\right ]\!\!,Z\!\!<\!\!\frac{E_{th}d_{1}^{\alpha }}{\eta \xi \tau P_{t}(1\!\!+\!\!\beta )}\!\! \right )\\&
=\mathbb{E}_{Z}\left\{ exp\left ( - d_{1}^{\alpha }\!\!\left [\! \frac{E_{th}}{\eta \xi \tau P_{t}}\!-\!d_{1}^{-\alpha }Z \!\right ]\!\!,Z\!\!<\!\!\frac{E_{th}d_{1}^{\alpha }}{\eta \xi \tau P_{t}(1\!\!+\!\!\beta )} \right ) \right\}\\&
=\int_{z=0}^{\frac{E_{th}d_{1}^{\alpha }}{\eta \xi \tau P_{t}(1+\beta )}}exp\left ( -d_{1}^{\alpha }\left [ \frac{E_{th}}{\eta \xi \tau P_{t}}-d_{1}^{-\alpha }Z \right ] \right )f_{Z}(z)d_{z}
\label{eq36}
\end{aligned}    
\end{equation}
The PDF of $Z$, $f_{Z}(Z)$ is as follows: $f_{Z}(Z)=\frac{e^{-Z}Z^{k-1}}{(k-1)!}$. Using $f_{Z}(Z)$ in \eqref{eq36}, we get
\begin{equation}
\begin{aligned}
R_{5}=\int_{z=0}^{\frac{E_{th}d_{1}^{\alpha }}{\eta \xi \tau P_{t}(1+\beta )}}\frac{Z^{k-1}}{(k-1)!}exp\left ( -\left [ \frac{E_{th}d_{1}^{\alpha }}{\eta \xi \tau P_{t}} \right ] \right )dz
\label{eq37}
\end{aligned}    
\end{equation}
To solve $R_{6}$, we can write
\begin{equation}
\begin{aligned}
R_{6}& = \mathbb{P}\left [ g_{1}^{j}>\beta Z,Z>\frac{E_{th}d_{1}^{\alpha }}{\eta \xi \tau P_{t}(1+\beta )} \right ]\\&
= \mathbb{E}_{Z}\left\{exp\left ( -\beta Z\right ),Z>\frac{E_{th}d_{1}^{\alpha }}{\eta \xi \tau P_{t}(1+\beta )} \right\}\\&
= \int_{z=\frac{E_{th}d_{1}^{\alpha }}{\eta \xi \tau P_{t}(1+\beta )}}^{\infty }exp\left ( -\beta Z \right )f_{Z}(z)dZ
\label{eq38}
\end{aligned}   
\end{equation}
The PDF of $Z$, $f_{Z}(Z)$ is as follows: $f_{Z}(Z)=\frac{e^{-Z}Z^{k-1}}{(k-1)!}$. Using $f_{Z}(Z)$ in \eqref{eq38}, we get
\begin{equation}
\begin{aligned}
R_{6} = \int_{z=\frac{E_{th}d_{1}^{\alpha }}{\eta \xi \tau P_{t}(1+\beta )}}^{\infty }\frac{Z^{k-1}}{(k-1)!}exp\left ( -(\beta+1) Z \right )dZ    
\label{eq39}
\end{aligned}    
\end{equation}
Replacing \eqref{eq37} and \eqref{eq39} in \eqref{eq33}, we get \eqref{theorem1eqC}
\section{Proof of Theorem 4} \label{AppendixD}
We can obtain 
$\mathbb{E}\big [ W_{i} \big ]\!=\!\sum\limits_{n=1}^{\infty }n\big [\mu _{\phi ,L}\big (1- \mu _{\phi ,L}\big )^{n-1} \big ]$.
Replacing $1- \mu _{\phi ,L} = t$, we get,\\
\begin{equation}
\begin{aligned}
&\mathbb{E}\big [ W_{i} \big ]=\!\sum\limits_{n=1}^{\infty }n\big ( 1\!-\!t \big )t^{n-1}= \frac{1}{1-t}= \frac{1}{ \mu _{\phi ,L}}.
\label{e33}
\end{aligned}
\end{equation}
Similarly for $V_{i}$ we can obtain
 $\mathbb{E}\big [ V_{i} \big ]= \sum_{n=0}^{\infty }n\big [p_{a}\big (1-p_{a}\big )^{n} \big ].$
Replacing $1-p_{a} = w$, we get,
\begin{equation}
\begin{aligned}
\mathbb{E}\big [ V_{i} \big ]= \sum_{n=0}^{\infty }n\big ( 1-w \big )w^{n}= \frac{w}{1-w}  \approx Z_{a},
\label{e35}
\end{aligned}
\end{equation}
Thus, the upper bound of \ac{PAoI} measured at the typical IoT is given by  
$A_{i,NP}^{U}= \mathbb{E}\big [ A_{i}|(\beta ,E_{th}):\phi  \big ]$. From \eqref{Yi} and \eqref{e2},
$A_{i,NP}^{U}= \mathbb{E}\big [ W_{i-1} \big ] + \mathbb{E}\big [V_{i}  \big ] + \mathbb{E}\big [W_{i}  \big ]$.
From \eqref{e33} and \eqref{e35}, we get \eqref{e16}.
\vspace{-0.2cm}
\section{Proof of Theorem 5} \label{AppendixE}
There are $N_{i}=n$ packet arrivals in (s-m) slots and  no packet arrivals in m slots.\\
\begin{equation}
\begin{aligned}
\mathbb{P}\!\left [ \hat{W_{i}}\!\!=m/W_{i}\!\!=\!\!s \right ]\!\!=&\mathbb{P}\!\left [ N_{i}\!\!=\!\textup{n packet arrivals in (s-m) slots }\!\!\right ]\\&*\mathbb{P}\left [ \textup{no packet arrivals in m slots} \right ]\\
=&\sum_{n=0}^{s-m}\mathbb{P}\left [ X_{n} \right ]*\left ( 1-p_{a} \right )^{m}\\
\end{aligned}    
\end{equation} 
We want to know the number of trials $(s-m)$ required to achieve a fixed number of successfully arrived replacement packets, n. where, n= 0 to (s-m), using negative binomial distribution we get,
\begin{equation}
\begin{aligned}
=\sum_{n=0}^{s-m}\binom{(s-m)-1}{n-1}p_{a}^{n}\left ( 1-p_{a} \right )^{(s-m)-n}*\left (1-p_{a}  \right )^{m}    
\end{aligned}    
\end{equation}
Putting n=1 because we have at least one arrived packet that is being serviced in the consecutive time slots gives,
\begin{equation}
\begin{aligned}
=\left\{\begin{matrix}
p_{a} &,(s=1)  \\
\sum_{n=1}^{(s-m)+1}\binom{s-m}{n-1}p_{a}^{n}\left ( 1-p_{a} \right )^{s-n} & ,(s\geq m> 1) \\
\left ( 1-p_{a} \right )^{m-1} & ,(s=m) \\
\end{matrix}\right.
\end{aligned}    
\end{equation}
PMF of $\hat{W_{i}}$ is as follows:
\begin{equation}
\begin{aligned}
\mathbb{P}&\left [  \hat{W_{i}}=m \right ]\\&=\sum_{s=m}^{\infty }\mathbb{P}\left [ \hat{W_{i}}=m/W_{i}=s \right ]\mathbb{P}\left [ W_{i}=s \right ]\\&
=\mathbb{P}\left [ \hat{W_{i}}=m/W_{i}=m \right ]\mathbb{P}\left [ W_{i}=m \right ] +\\&\sum_{s=m+1}^{\infty } \mathbb{P}\left [ \hat{W_{i}}=m/W_{i}=s \right ]\mathbb{P}\left [ W_{i}=s \right ]\\&
=\left ( 1-p_{a} \right )^{m-1}\mu _{\phi ,L}\left ( 1-\mu _{\phi ,L} \right )^{m-1}+\\&\underset{G}{\underbrace{\sum_{n=1}^{(s-m)+1}\!\!\!\!\binom{s-m}{n-1}p_{a}^{n}\left ( 1-p_{a} \right )^{s-n}\!\!\!\!\sum_{s=m+1}^{\infty }\!\!\!\!\mu _{\phi ,L}\left ( 1-\mu _{\phi ,L} \right )^{s-1}}}
\end{aligned}    
\end{equation}
Substituting $s-m=l$ in the above term G we get,
\begin{equation}
\begin{aligned}
G= \sum_{n=1}^{l+1}\!\!\binom{l}{n-1}p_{a}^{n}\left ( 1-p_{a} \right )^{l+m-n}\!\sum_{l=1}^{\infty }\mu _{\phi ,L}\left ( 1-\mu _{\phi ,L} \right )^{l+m-1}
\end{aligned}    
\end{equation}
Put $l+1=z$ and $n=1$, we get,
\begin{equation}
\begin{aligned}
G= p_{a}\left ( 1-p_{a} \right )^{m-1}\mu _{\phi ,L}\left ( 1-\mu _{\phi ,L} \right )^{m}\sum_{z=2}^{\infty }\left ( 1-\mu _{\phi ,L} \right )^{z-2}
\end{aligned}    
\end{equation}
Put $z-2=a$ and $\sum_{a=0}^{\infty }x^{a}=\frac{1}{1-x}$ for $|x|< 1$
\begin{equation}
\begin{aligned}
G= p_{a}\left ( 1-p_{a} \right )^{m-1}\left ( 1-\mu _{\phi ,L} \right )^{m}
\end{aligned}    
\end{equation}
Substituting G in $\mathbb{P}\left [  \hat{W_{i}}=m \right ]$ we get,
\begin{equation}
\begin{aligned}
\mathbb{P}\left [  \hat{W_{i}}=m \right ]=&\left ( 1-p_{a} \right )^{m-1}\mu _{\phi ,L}\left ( 1-\mu _{\phi ,L} \right )^{m-1}+\\&p_{a}\left ( 1-p_{a} \right )^{m-1}\left ( 1-\mu _{\phi ,L} \right )^{m}\\
=&q_{s}\left ( 1-q_{s} \right )^{m-1} \textup{for m = 1, 2....} 
\end{aligned}    
\end{equation}
where, $q_{s}=\mu _{\phi ,L}+p_{a}\left ( 1-\mu _{\phi ,L} \right )$
We can obtain $\mathbb{E}\left [ \hat{W_{i}} \right ]=\sum_{m=1}^{\infty }m\left [ q_{s}\left ( 1-q_{s} \right )^{m-1} \right ]$. Replacing $1-q_{s}=t$, we get,
\begin{equation}
\begin{aligned}
\mathbb{E}\left [ \hat{W_{i}} \right ]=\sum_{m=1}^{\infty }m\left [ q_{s}\left ( 1-q_{s} \right )^{m-1} \right ]=\frac{1}{1-t}=\frac{1}{q _{s}}
\label{e53p}
\end{aligned}    
\end{equation}
Similarly for $V_{i}$ we can obtain
 $\mathbb{E}\big [ V_{i} \big ]= \sum_{n=0}^{\infty }n\big [p_{a}\big (1-p_{a}\big )^{n} \big ].$
Replacing $1-p_{a} = w$, we get,
\newpage
\begin{equation}
\begin{aligned}
\mathbb{E}\big [ V_{i} \big ]= \sum_{n=0}^{\infty }n\big ( 1-w \big )w^{n}= \frac{w}{1-w}  \approx Z_{a},
\label{e54p}
\end{aligned}
\end{equation}

We can obtain $\mathbb{E}\left [W_{i} \right ]=\sum_{m=1}^{\infty }m\left [\mu _{\phi ,L}\left ( 1-\mu _{\phi ,L} \right )^{m-1} \right ]$. Replacing $1-\mu _{\phi ,L}=v$, we get,
\begin{equation}
\begin{aligned}
\mathbb{E}\left [W_{i} \right ]=\sum_{m=1}^{\infty }m\left [\mu _{\phi ,L}\left ( 1-\mu _{\phi ,L} \right )^{m-1} \right ]=\frac{1}{1-v}=\frac{1}{\mu _{\phi ,L}}
\label{e55p}
\end{aligned}    
\end{equation}

Thus, the upper bound of \ac{PAoI} measured at the typical IoT is given by  
$A_{i,P}^{U}= \mathbb{E}\big [ A_{i}|(\beta ,E_{th}):\phi  \big ]$. From \eqref{Yip} and \eqref{e2p},
$A_{i,P}^{U}= \mathbb{E}\big [\hat{ W_{i-1}} \big ] + \mathbb{E}\big [V_{i}  \big ] + \mathbb{E}\big [W_{i}  \big ]$.
From \eqref{e53p}, \eqref{e54p} and \eqref{e55p}, we get \eqref{e16p}.

\bibliographystyle{IEEEtran}
\bibliography{IEEEabrv,references}
\end{document}